\documentclass[11pt,a4paper]{article}
\usepackage{amsmath,amssymb,latexsym}
\usepackage{graphicx, pdfpages}
\usepackage [utf8]{inputenc}
\usepackage[T2A]{fontenc}
\usepackage[left=2cm,right=2cm, top=1cm,bottom=1cm]{geometry}
\usepackage[colorlinks=true]{hyperref}
\usepackage{amsthm}
\usepackage{enumitem}
\usepackage{bm,bbm}
\usepackage{xcolor}

\numberwithin{equation}{section}

\newcommand{\Compl}{\mathbb{C}}
\newcommand{\R}{\mathbb{R}}

\newcommand{\conj}[1]{\overline{#1}}
\newcommand{\cConjScl}[1]{\bar{#1}}

\newcommand{\aConjScl}[1]{#1^*}
\newcommand{\aConjMat}[1]{#1^+}

\newcommand{\charfp}[2]{\psi\left(#1,#2\right)}

\newcommand{\CF}{\mathsf{f}}
\newcommand{\qform}{\xi}

\newcommand{\ens}{M_n}
\newcommand{\sclA}{y}
\newcommand{\matA}{Y}

\newcommand{\cMatGin}{X}
\newcommand{\cMatPos}{\mathcal{Z}}
\newcommand{\cSclN}{t}
\newcommand{\cVecN}{\bm{\cSclN}}
\newcommand{\cSclA}{q}
\newcommand{\cMatA}{Q}

\newcommand{\cSclB}{a}
\newcommand{\cVecGI}{\bm{h}}
\newcommand{\EcVec}{\bm{u}}
\newcommand{\UcVec}{\bm{u}}
\newcommand{\QcVec}{\bm{q}}

\newcommand{\aSclB}{\phi}
\newcommand{\aVecB}{\bm{\aSclB}}
\newcommand{\aMatB}{\Phi}
\newcommand{\aSclBt}{\varphi}
\newcommand{\aVecBt}{\bm{\aSclBt}}
\newcommand{\aSclC}{\theta}
\newcommand{\aVecC}{\bm{\aSclC}}
\newcommand{\aMatC}{\Theta}
\newcommand{\aSclCt}{\vartheta}
\newcommand{\aVecCt}{\bm{\aSclCt}}
\newcommand{\aSclD}{\rho}
\newcommand{\aVecD}{\bm{\aSclD}}
\newcommand{\aSclE}{\tau}

\newcommand{\aSclH}{\aSclE}
\newcommand{\aVecH}{\bm{\aSclH}}
\newcommand{\aVecGrI}{\bm{\upsilon}}

\newcommand{\cumul}[2]{\kappa_{#1,#2}}
\newcommand{\indexset}{\mathcal{I}}

\newcommand{\stpointsnbh}{\Omega_n}
\newcommand{\idom}{\mathcal{D}}
\newcommand{\Vanddet}{\triangle}
\newcommand{\herm}{\mathcal{H}}

\newcommand{\der}[2]{\frac{d #1}{d #2}}

\newcommand{\abs}[1]{\left\lvert#1\right\rvert}

\newcommand{\norm}[1]{\left\lVert#1\right\rVert}
\newcommand{\normsized}[2][ ]{#1\lVert#2#1\rVert}

\DeclareMathOperator{\tr}{tr}
\DeclareMathOperator{\E}{\mathbf{E}}

\DeclareMathOperator{\diag}{diag}

\newtheorem{thm}{Theorem}
\newtheorem{prop}{Proposition}
\newtheorem{lem}{Lemma}

\theoremstyle{remark}

\title{Characteristic polynomials of sparse non-Hermitian random matrices}
\author{Ievgenii Afanasiev\thanks{B. Verkin Institute for Low Temperature Physics and Engineering of the National Academy of Sciences of Ukraine, Kharkiv, Ukraine;
e-mail: afanasiev@ilt.kharkov.ua, ie.afanasiev@gmail.com. The author is partially supported by the Grant EFDS-FL2-08 of the found The European Federation of Academies of Sciences and Humanities (ALLEA).} \and
 Tatyana Shcherbina
\thanks{ Department of Mathematics, University of Wisconsin - Madison, USA, e-mail: tshcherbyna@wisc.edu. Supported in part by Alfred P. Sloan Foundation grant FG-2022-18916.}
}
\date{}

\begin{document}

\maketitle

\begin{abstract}
We consider the asymptotic local behavior of the second correlation functions of the characteristic polynomials of sparse non-Hermitian random matrices $X_n$ whose entries have the form $x_{jk}=d_{jk}w_{jk}$ with iid complex standard Gaussian $w_{jk}$ and normalised iid Bernoulli$(p)$ $d_{jk}$. It is shown that, as $p\to\infty$, the local asymptotic behavior of the second correlation function of characteristic polynomials near $z_0\in \mathbb{C}$ coincides with those for Ginibre ensemble: it converges to a determinant with Ginibre kernel in the bulk $|z_0|<1$, and it is factorized if $|z_0|>1$. For the finite $p>0$, the behavior is different and exhibits the transition between different regimes depending on values $p$ and $|z_0|^2$.    
\end{abstract}

\section{Introduction}
Introduce $n\times n$ non-Hermitian random matrices 
\begin{equation} \label{sparse}
 X_n = (x_{jk})_{j,k = 1}^n,
\end{equation}
whose entries can be written in the form $$x_{jk} = d_{jk}w_{jk}$$ 
with i.i.d. complex random  variables $w_{jk}$ such that
\begin{align}\label{moments}
\E\{w_{jk}\} = \E\{w_{jk}^2\} = 0, \quad \E\{\abs{w_{jk}}^2\} = 1,
\end{align}
 and  normalized   i.i.d. Bernoulli$(p)$, $0<p\le n$, indicators $d_{jk}$ independent of $\{w_{jk}\}$, i.e. 
\begin{align}\label{d}
d_{jk} = \frac{1}{\sqrt{p}}\begin{cases}
    1, \text{ with probability } \frac{p}{n},\\
    0, \text{ with probability } 1 - \frac{p}{n}.
\end{cases}
\end{align}
We will refer to this ensemble as to {\it sparse non-Hermitian random matrices}. Parameter $p$ can be fix or be dependent on $n$.
Clearly, this ensemble interpolates between non-Hermitian random matrices with iid entries for $p=n$ and a very sparse matrices (who has, on average,
$p$ non-zero entries in a row) for a fixed $p$. 

The  limiting empirical spectral distributions for such matrices with $p$ growing together with $n$ 
was excessively studied in the mathematical literature (see, e.g., \cite{TV:08}, \cite{GoTi:10}, \cite{W:12}, \cite{BR:19} and references therein)
with an optimal result obtained by Rudelson and Tikhomirov in \cite{RT:19}: as soon as $p\to\infty$ together with $n$, the empirical spectral distribution of
sparse non-Hermitian random matrices converges  weakly in probability to the circular law, i.e. the uniform distribution on a unit disk.
The existence of the limiting empirical spectral distributions for finite $p>0$ (in this case it is not a circular law anymore) was obtained very 
recently in \cite{SSS:23}.

The local eigenvalue statistics of (\ref{sparse}) is much less studied. For a non-Hermitian matrices with iid random entries (i.e. $p=n$ case) the local
eigenvalue statistics in the bulk and at the edge of the spectrum coincide with those of the Ginibre ensemble, i.e. matrices with iid Gaussian entries. This is known as the {\it universality of non-Hermitian random matrices} (see \cite{Ta-Vu:15}, \cite{CiErS:ed} and references therein). For the sparse case, the universality at the edge of the spectrum for $ n^{\alpha}\ll p \le \tfrac{1}{2}$ was obtained recently in \cite{He:23}. The bulk universality with $p\ll n$ is still an open question.

In this paper we are going to study the local behavior of correlation functions of characteristic polynomials. For the non-Hermitian random matrices it can be defined as
\begin{equation}\label{F_m}
\CF_k (z_1,\ldots, z_k)=\mathbf{E}\Big\{\prod\limits_{s=1}^{k}\det(H_n-z_s)\det(H_n-z_s)^*\Big\}.
\end{equation}
We are interested in the asymptotic behavior of $\CF_2$ for matrices (\ref{sparse}) as $n\to\infty$
and
\begin{equation}\label{z}
 z_j=z_0+\tfrac{\zeta_j}{\sqrt n},\,\,j=1,2.
\end{equation}
Characteristic polynomials are the objects of independent interest because of their connections to the
number theory, quantum chaos, integrable systems, combinatorics, representation theory and others.
In additional, although $\CF_k$ is not a local object in terms of eigenvalue statistics, it is also expected 
to be universal in some sense. In particular, it was proved in \cite{Af:19} (see also \cite{Ak-Ve:03} for  the Gaussian (Ginibre) case)
that for non-Hermitian random matrices $H$ with iid complex entries with mean zero, variance one, and $2k$ finite moments
for any $z_j=z_0+\zeta_j/\sqrt n$, $j=1,..,k$ and $|z_0|<1$ we get
\begin{equation}\label{ChP_lim}
\lim\limits_{n\to\infty} n^{-\tfrac{k^2-k}{2}}\dfrac{\CF_{k}(z_1,\ldots, z_k)}{\prod_j \CF_1(z_j)}=C_k \dfrac{\det(K^{(b)}(\zeta_i,\zeta_j))_{i,j}^k}{|\Delta(\zeta)|^2},
\end{equation}
where 
\begin{equation}\label{K_b}
K^{(b)}(w_1,w_2)=e^{-|w_1|^2/2-|w_2|^2/2+w_1\bar w_2},
\end{equation}
$\Delta(\zeta)$ is  a Vandermonde determinant of $\zeta_1,\ldots,\zeta_k$ and
$C_k$ is constant depending only on the fourth cumulant $\kappa_4=\E[|H_{11}|^4]-2$ of the elements distribution, but not on the higher moments.
In particular, this means that the local limiting behavior  (\ref{ChP_lim}) of non-Hermitian matrices with iid entries coincides with those for the Ginibre ensemble as soon as the elements distribution has four Gaussian moments, i.e. the local behavior of the correlation functions of characteristic polynomials also exhibits a certain form of universality. Similar results were obtained for many classical Hermitian random matrix ensembles
(see, e.g., \cite{Br-Hi:00}, \cite{Br-Hi:01}, \cite{St-Fy:03}, \cite{TSh:ChW}, \cite{TSh:ChSC}, \cite{Af:16}, \cite{TSh:ChB}, etc.)

Notice that the local asymptotic behavior of characteristic polynomials of the sparse {\it Hermitian} random matrices was obtained in \cite{Af:16}.
In particular, it was shown that while for $p\to\infty$ the behavior coincides with that for Gaussian Unitary Ensemble (GUE), i.e. Hermitian matrices with iid (up to the symmetry) Gaussian entries.  For the finite $p$ the local asymptotic behavior
 of the second correlation function of characteristic polynomials of spares Hermitian random matrices demonstrates the transition: when $p<2$ the second correlation function of
 characteristic polynomials factorizes in the
limit $n\to\infty$, while for $p>2$ there appears an interval $(\lambda_-(p), \lambda_+(p))$ such that inside $(\lambda_-(p), \lambda_+(p))$ the second correlation function behaves like that for GUE, while
outside the interval the second correlation function is still factorized.

The goal of the current paper is to establish similar result for the sparse non-Hermitian matrices (\ref{sparse}). 
Define
\begin{equation}\label{b}
    b = \sqrt{\tfrac{2(n - p)}{np}}
\end{equation}
with $p$ of (\ref{d}). Notice that if $p$ is finite, then
\begin{equation}\label{fin_p}
    b=\sqrt{\dfrac{2}{p}}+O(n^{-1}),\quad n\to\infty,
\end{equation}
and 
\begin{equation}\label{inf_p}
    b=O(p^{-1/2}),\quad n\to\infty,
\end{equation}
if $p\to\infty$ but $p<(1-\varepsilon)n$ for some $\varepsilon>0$.

\begin{figure}
\centering
\includegraphics[width=0.7\textwidth]{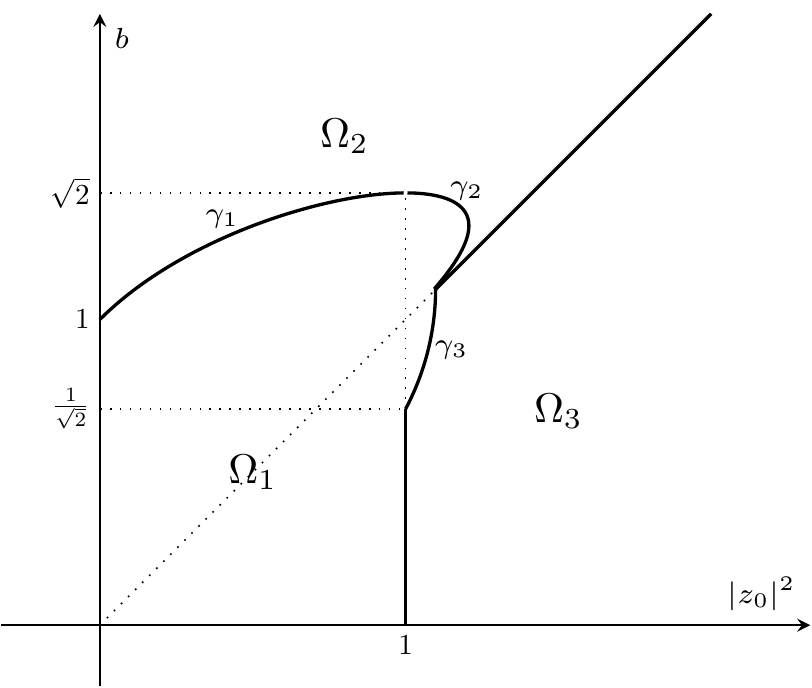}
\caption{phase diagram of the three different types of  behavior of $\CF_2$}
\end{figure}

The main result of the paper is the following theorem:
\begin{thm}\label{thm:main}
   Let $X_n$ be the sparse non-Hermitian complex random matrices (\ref{sparse}) with the standard complex Gaussian $w_{jk}$ and finite fixed $p>0$.
   Then for the second correlation function of characteristic polynomials $\CF_2$ of (\ref{F_m}) with $z_1, z_2$ from (\ref{z}) we have 
   \begin{enumerate}[label=(\roman*)]
       \item if $(b,|z_0|^2)\in \Omega_1$, then 
       \begin{equation}\label{*-behav}
           \lim_{n \to \infty} \dfrac{\CF_{2}(z_1, z_2)}{\sqrt{\CF_2(z_1,z_1) \CF_2(z_2,z_2)}} = e^{ -\gamma(\Re (\bar z_0(\zeta_1 - \zeta_2)))^2} 
           \dfrac{\det(K^{(b)}(\sqrt{\beta}\zeta_i,\sqrt{\beta}\zeta_j))_{i,j}^2}{\beta\abs{\Vanddet(\zeta)}^2},
       \end{equation}
       where $K^{(b)}$ is defined in (\ref{K_b}) and $\beta \in [0, 1]$ is a solution to the equation
       \begin{equation*}
       p\beta - p + 2 = p\abs{z_0}^2(1 - \beta)^2(2 - p\abs{z_0}^2(1 - \beta)),
       \end{equation*}
       and $\gamma>0$ is a certain constant depending on $p$ and $|z_0|^2$ (see (\ref{gamma})).
       \item if $(b,|z_0|^2)\in \Omega_2$, then 
       \begin{equation}\label{v-behav}
           \lim_{n \to \infty} \dfrac{\CF_{2}(z_1, z_2)}{\sqrt{\CF_2(z_1,z_1) \CF_2(z_2,z_2)}} = e^{ -p(\Re( \bar z_0(\zeta_1 - \zeta_2)))^2}.
       \end{equation}
       
       \item if $(b,|z_0|^2)\in\Omega_3$, then 
       \begin{equation}\label{0-behav}
           \lim_{n \to \infty} \dfrac{\CF_{2}(z_1, z_2)}{\sqrt{\CF_2(z_1,z_1) \CF_2(z_2,z_2)}} = 1.
       \end{equation}
   \end{enumerate}
The domains $\Omega_i$, $i=1,2,3$ are shown in Figure 1. Here $\gamma_1$ is a curve
$$|z_0|^2=\dfrac{b^2-b\sqrt{2-b^2}}{2}, \quad b\in [1,\sqrt2],$$  
and $\gamma_2$, $\gamma_3$ are certain explicit curves that will be defined later (see Lemma \ref{l:I>II} and Lemma \ref{l:I>III}).
\end{thm}
In the case $p\to \infty$ we can also get
\begin{thm}\label{thm:inf}
   Let $X_n$ be the sparse non-Hermitian complex random matrices (\ref{sparse}) with the standard complex Gaussian $w_{jk}$ and $p\to\infty$ but $p<(1-\varepsilon)n$ for some $\varepsilon>0$. Then the limiting behavior of second correlation function of characteristic polynomials $\CF_2$ of (\ref{F_m}) with $\{z_j\}_{j=1,2}$ from (\ref{z}) coincides with those for Ginibre ensemble (\ref{ChP_lim}) in the bulk, i.e.  if $|z_0|<1$. If $|z_0|>1$, then (\ref{0-behav}) holds.
\end{thm}
In order to prove Theorems \ref{thm:main}-\ref{thm:inf} we are going to apply the supersymmetry techniques (SUSY). SUSY techniques is based on the representation of the determinant as an integral (formal) over the Grassmann variables, which allows to obtain the integral representation for
 the main spectral characteristics of random matrices (such as density of states, correlation functions, characteristic polynomial, etc.)  as an integral containing both
complex and Grassmann (anticommuting) variables.  Although at a heuristic level  SUSY 
was actively  used in theoretical physics literature (see e.g. reviews \cite{Ef:97}, \cite{M:00})
for several decades, the rigorous analysis of such integrals poses a very serious challenge.
However, the method was successfully applied to the rigorous study of local regime of some random matrix ensembles, including the most
successful applications to the Gaussian Hermitian random band matrices (see \cite{SS:Un} and reference therein).
The asymptotic behavior of characteristic polynomials is known to be especially convenient 
for the SUSY approach and were successfully studied by the techniques
for many Hermitian (see, e.g., \cite{Br-Hi:00}, \cite{Br-Hi:01}, \cite{TSh:ChW}, \cite{TSh:ChSC},\cite{Af:16}, \cite{TSh:ChB} etc.) and some
non-Hermitian (see, e.g.,  \cite{APSo:09}, \cite{Af:19}, \cite{Af:22}) ensembles. 

The paper organized as follows. In Section 2 we give the brief outline of SUSY techniques and obtain the SUSY integral representation of $\CF_2$ of (\ref{F_m}). Section 3 is devoted to the proof of Theorem \ref{thm:main} by performing the saddle-point analysis of the obtained representation: in Section 3.1 we determine the main saddle-points that can give the leading contribution to the integral; in Section 3.2 we determine the domain of domination of
each of the obtained saddle-point; and finally the contribution of each of the saddle-points is computed in  Section 3.3.

\section{Integral representation for $\CF_2$}
In this section we obtain a convenient integral representation for the correlation function of characteristic polynomials $\CF_2$ defined by \eqref{F_m}.
\begin{prop}\label{prop:IR}
	Let an ensemble $\ens$ be defined by~\eqref{sparse} and~\eqref{moments}. Then 
	the second correlation function of the characteristic polynomials $\CF_2$ defined by \eqref{F_m} can be represented in the following form
	\begin{equation}\label{IR result}
	\CF_2 (z_1,z_2) = \left(\frac{n}{\pi}\right)^{5} \int e^{n\hat{f}(Q,v)} dQ\, d\bar v\, dv,
	\end{equation}
	where $\cMatA$ is a complex $2 \times 2$ matrix, $v\in\Compl$, $$d Q = \prod\limits_{j,k = 1}^2 d \bar Q_{jk}\, dQ_{kj} $$ and
	\begin{align}
		\label{tilde f def}
		&\hat{f}(Q,v) = - \tr Q^*Q - |v|^2 + \log h(Q,v); \\
		\label{h def}
		&h(Q,v) = \det A + bv\det Q^* + b\bar v\det Q + b^2 \abs{v}^2; \\
		\label{A def}
		&A = A(Q) = \begin{pmatrix}
			-Z 			& \cMatA\\
			-\cMatA^* 	& -Z^*
		\end{pmatrix}, \quad Z=\begin{pmatrix}
     z_1&0\\
     0&z_2
 \end{pmatrix}
	\end{align}
	with $b$ of (\ref{b}). 
\end{prop}
%


\subsection{Proof of Proposition \ref{prop:IR}}

To derive the integral representation of $\CF_2$ we will use the SUSY. For the reader convenience, we start with a very brief outline of the basic formulas of SUSY techniques we need. More detailed information about the techniques and its applications 
to random matrix theory  can be found, e.g.,
in \cite{Ef:97} or \cite{M:00}.

\subsection{SUSY techniques: basic formulas}
Let us consider two sets of formal variables
$\{\psi_j\}_{j=1}^n,\{\overline{\psi}_j\}_{j=1}^n$, which satisfy the anticommutation relations
\begin{equation}\label{anticom}
\psi_j\psi_k+\psi_k\psi_j=\overline{\psi}_j\psi_k+\psi_k\overline{\psi}_j=\overline{\psi}_j\overline{\psi}_k+
\overline{\psi}_k\overline{\psi}_j=0,\quad j,k=1,\ldots,n.
\end{equation}
Note that this definition implies $\psi_j^2=\overline{\psi}_j^2=0$.
These two sets of variables $\{\psi_j\}_{j=1}^n$ and $\{\overline{\psi}_j\}_{j=1}^n$ generate the Grassmann
algebra $\mathfrak{A}$. Taking into account that $\psi_j^2=0$, we have that all elements of $\mathfrak{A}$
are polynomials of $\{\psi_j\}_{j=1}^n$ and $\{\overline{\psi}_j\}_{j=1}^n$ of degree at most one
in each variable. We can also define functions of
the Grassmann variables. Let $\chi$ be an element of $\mathfrak{A}$, i.e.
\begin{equation}\label{chi}
\chi=a+\sum\limits_{j=1}^n (a_j\psi_j+ b_j\overline{\psi}_j)+\sum\limits_{j\ne k}
(a_{j,k}\psi_j\psi_k+
b_{j,k}\psi_j\overline{\psi}_k+
c_{j,k}\overline{\psi}_j\overline{\psi}_k)+\ldots.
\end{equation}
For any
sufficiently smooth function $f$ we define by $f(\chi)$ the element of $\mathfrak{A}$ obtained by substituting $\chi-a$
in the Taylor series of $f$ at the point $a$:
\[
f(\chi)=a+f'(a)(\chi-a)+\dfrac{f''(a)}{2!}(\chi-a)^2+\ldots
\]
Since $\chi$ is a polynomial of $\{\psi_j\}_{j=1}^n$,
$\{\overline{\psi}_j\}_{j=1}^n$ of the form (\ref{chi}), according to (\ref{anticom}) there exists such
$l$ that $(\chi-a)^l=0$, and hence the series terminates after a finite number of terms and so $f(\chi)\in \mathfrak{A}$.

Following Berezin \cite{Be:87}, we define the operation of
integration with respect to the anticommuting variables in a formal
way:
\begin{equation*}
\int d\,\psi_j=\int d\,\overline{\psi}_j=0,\quad \int
\psi_jd\,\psi_j=\int \overline{\psi}_jd\,\overline{\psi}_j=1,
\end{equation*}
and then extend the definition to the general element of $\mathfrak{A}$ by
the linearity. A multiple integral is defined to be a repeated
integral. Assume also that the ``differentials'' $d\,\psi_j$ and
$d\,\overline{\psi}_k$ anticommute with each other and with the
variables $\psi_j$ and $\overline{\psi}_k$. Thus, according to the definition, if
$$
f(\psi_1,\ldots,\psi_k)=p_0+\sum\limits_{j_1=1}^k
p_{j_1}\psi_{j_1}+\sum\limits_{j_1<j_2}p_{j_1,j_2}\psi_{j_1}\psi_{j_2}+
\ldots+p_{1,2,\ldots,k}\psi_1\ldots\psi_k,
$$
then
\begin{equation*}
\int f(\psi_1,\ldots,\psi_k)d\,\psi_k\ldots d\,\psi_1=p_{1,2,\ldots,k}.
\end{equation*}

   The key formulas we need in this subsection are the well-known complex Gaussian integration formula for a complex $n$-dimensional vector $\cVecN$
\begin{equation} \label{Gauss int}
\int\limits_{\Compl^n} \exp\left\{-\cVecN^*B\cVecN - \cVecN^*\cVecGI_2 - \cVecGI_1^*\cVecN\right\} d\cVecN^*d\cVecN = 
\pi^n {\det}^{-1} B \exp \{\cVecGI_1^*B^{-1}\cVecGI_2\},
\end{equation}
valid for any positive definite matrix $B$ and its analog for
Grassmann $n$-dimensional vector $\aVecH$ (see \cite{Be:87}):
\begin{equation}
	\label{Grass int}
	\int \exp\left\{-\aConjMat{\aVecH}A\aVecH - \aConjMat{\aVecH}\aVecGrI_2 - \aConjMat{\aVecGrI_1}\aVecH\right\} d\aConjMat{\aVecH}d\aVecH = 
	\det B \exp \{\aConjMat{\aVecGrI_1}B^{-1}\aVecGrI_2\} 
\end{equation}
valid for an arbitrary complex matrix $B$.

We will need also the following Hubbard--Stratonovich transformation formula which basically is an employment of \eqref{Gauss int} in the reverse direction:
\begin{align}\label{Hub_C}
&e^{ab}=\pi^{-1} \int\limits_{\mathbb{C}} e^{a \bar u+b u-\bar u u} d\bar u\, du.
\end{align}
Here $a, b$ can be complex numbers or sums of the products of even numbers of Grassmann variables (i.e. commuting elements of Grassmann algebra).

\subsection{Integral representation}
 Rewrite the expression \eqref{F_m} for $\CF_2$ using~\eqref{Grass int} and~\eqref{sparse}
\begin{align*}
	\CF_2(z_1,z_2) = \E \Bigg\{
	\int \exp \Bigg\{ -\sum\limits_{j = 1}^{2} \aConjMat{\aVecB_j}\left(\cMatGin - z_j\right)\aVecB_j - \sum\limits_{j = 1}^{2} \aConjMat{\aVecC_j}\left(\cMatGin - z_j\right)^*\aVecC_j \Bigg\} d\aMatB d\aMatC
	\Bigg\},
\end{align*}
where $\aVecB_j$, $\aVecC_j$, $j = 1, 2$ are $n$-dimensional vectors with components $\aSclB_{kj}$ and $\aSclC_{kj}$ respectively, $\aMatC$ and $\aMatB$ are $n\times 2$ matrices composed of columns $\aVecC_1, \aVecC_2$ and $\aVecB_1, \aVecB_2$, and
\begin{equation*}
   d\aMatB = \prod\limits_{j = 1}^{2} d\aConjMat{\aVecB_j} d\aVecB_j,\quad  d\aMatC = \prod\limits_{j = 1}^{2} d\aConjMat{\aVecC_j} d\aVecC_j.
\end{equation*}
 Denoting $\aVecBt_k = (\phi_{k1},\phi_{k2})^t$, $\aVecCt_k =  (\theta_{k1},\theta_{k2})^t$ we can rewrite the previous formula as
\begin{multline} \label{IR 1}
	\CF_2(z_1,z_2) = \E \Bigg\{
	\int \exp \Bigg\{ \sum\limits_{k = 1}^n \aConjMat{\aVecBt_k}Z\aVecBt_k + \sum\limits_{k = 1}^n \aConjMat{\aVecCt_k}Z^*\aVecCt_k \\
	+ \sum\limits_{k,l = 1}^n (\aMatB\aConjMat{\aMatB})_{lk}x_{kl} + \sum\limits_{k,l = 1}^n (\aMatC\aConjMat{\aMatC})_{kl}\bar{x}_{kl} \Bigg\} d\aMatB d\aMatC
	\Bigg\},
\end{multline}
where $Z$ is defined in (\ref{A def}).

Let us introduce a notation for a kind of ``Laplace--Fourier transform''
\begin{equation*}
	\charfp{t_1}{t_2} := \E \left\{
	e^{t_1x_{11} + t_2\bar{x}_{11}}
	\right\}.
\end{equation*}
Then the expectation in \eqref{IR 1} can be written in the following form
\begin{align*}
	\CF_2(z_1,z_2) = \int& \prod\limits_{k,l = 1}^n \charfp{(\aMatB\aConjMat{\aMatB})_{lk}}{(\aMatC\aConjMat{\aMatC})_{kl}} \\
	&\times \exp \Bigg\{ \sum\limits_{k = 1}^n \aConjMat{\aVecBt_k}Z\aVecBt_k + \sum\limits_{k = 1}^n \aConjMat{\aVecCt_k}Z^*\aVecCt_k \Bigg\} d\aMatB d\aMatC \\
	= \int& \exp \Bigg\{ \sum\limits_{k = 1}^n \aConjMat{\aVecBt_k}Z\aVecBt_k + \sum\limits_{k = 1}^n \aConjMat{\aVecCt_k}Z^*\aVecCt_k \\
	&\phantom{ \exp \Bigg\{} + \sum\limits_{k,l = 1}^n \log\charfp{(\aMatB\aConjMat{\aMatB})_{lk}}{(\aMatC\aConjMat{\aMatC})_{kl}} \Bigg\} d\aMatB d\aMatC.
\end{align*}
Expansion of $\log\Phi$ into series gives us
\begin{align}
	\notag
	\CF_2(z_1,z_2) = \int \exp \Bigg\{& \sum\limits_{k = 1}^n \aConjMat{\aVecBt_k}Z\aVecBt_k + \sum\limits_{k = 1}^n \aConjMat{\aVecCt_k}Z^*\aVecCt_k \\
	\label{IR 2}
	&+ \sum\limits_{k,l = 1}^n \sum\limits_{p,s = 0}^2 \frac{\cumul{p}{s}}{p!s!} \left((\aMatB\aConjMat{\aMatB})_{lk}\right)^p\left((\aMatC\aConjMat{\aMatC})_{kl}\right)^s \Bigg\} d\aMatB d\aMatC,
\end{align}
with
\begin{align*} 
	\cumul{p}{s} = \left.\frac{\partial^{p + s}}{\partial^p y_1 \partial^s y_2} \log\charfp{y_1}{y_2}\right|_{y_1 = y_2 = 0}.
\end{align*}
Using \eqref{sparse}--\eqref{d}, one can compute
\begin{equation}\label{cumul values}
\begin{split}
\cumul{0}{0} &= 0; \\
\cumul{1}{0} &= \conj{\cumul{0}{1}} = \E\{x_{11}\} = 0; \\
\cumul{2}{0} &= \conj{\cumul{0}{2}} = \E\{x_{11}^2\} - \E^2\{x_{11}\} = 0; \\
\cumul{1}{1} &= \E\{\abs{x_{11}}^2\} - \abs{\E\{x_{11}\}}^2 = \frac{1}{n};\\
\cumul{2}{2} &= \E\{\abs{x_{11}}^4\} - 2\E^2\{\abs{x_{11}}^2\} = \frac{2(n - p)}{pn^2}.
\end{split}
\end{equation}
Let us transform the terms in the exponent again. For non-zero terms with $p=s=1$ or $p=s=2$ one can write
\begin{align}
	\notag
	\sum\limits_{k,l = 1}^n &\left((\aMatB\aConjMat{\aMatB})_{lk}\right)^p\left((\aMatC\aConjMat{\aMatC})_{kl}\right)^s \\
	\notag
	&= \sum\limits_{k,l = 1}^n \Bigg(\sum\limits_{j = 1}^2 \aSclB_{lj}^{\phantom{+}}\aConjScl{\aSclB_{kj}}\Bigg)^p\Bigg(\sum\limits_{j = 1}^2 \aSclC_{kj}^{\phantom{+}}\aConjScl{\aSclC_{lj}}\Bigg)^s = p!s! \sum\limits_{k,l = 1}^n \sum\limits_{\substack{\alpha \in \indexset_{2, p} \\ \beta \in \indexset_{2, s}}} \prod\limits_{q = 1}^{p} \aSclB_{l\alpha_{q}}^{\phantom{+}}\aConjScl{\aSclB_{k\alpha_{q}}} \prod\limits_{r = 1}^{s}\aSclC_{k\beta_{r}}^{\phantom{+}}\aConjScl{\aSclC_{l\beta_{r}}} \\
	\notag
	&= (-1)^{p^2} p!s! \sum\limits_{k,l = 1}^n \sum\limits_{\substack{\alpha \in \indexset_{2, p} \\ \beta \in \indexset_{2, s}}} \prod\limits_{r = s}^{1} \aSclC_{k\beta_{r}}^{\phantom{+}} \prod\limits_{q = p}^{1} \aConjScl{\aSclB_{k\alpha_{q}}} \prod\limits_{q = 1}^{p} \aSclB_{l\alpha_{q}}^{\phantom{+}} \prod\limits_{r = 1}^{s} \aConjScl{\aSclC_{l\beta_{r}}} \\
	\label{sum order change}
	&= p!s! \sum\limits_{\substack{\alpha \in \indexset_{2, p} \\ \beta \in \indexset_{2, s}}} \Bigg(\sum\limits_{k = 1}^n (-1)^p\prod\limits_{r = s}^{1} \aSclC_{k\beta_{r}}^{\phantom{+}} \prod\limits_{q = p}^{1} \aConjScl{\aSclB_{k\alpha_{q}}}\Bigg) \Bigg(\sum\limits_{k = 1}^n \prod\limits_{q = 1}^{p} \aSclB_{k\alpha_{q}}^{\phantom{+}} \prod\limits_{r = 1}^{s} \aConjScl{\aSclC_{k\beta_{r}}}\Bigg),
\end{align}
where 
\begin{equation} \label{indexset def}
\indexset_{2, 1} = \{1,2\}, \quad \indexset_{2, 2} = \{\{1,2\}\}
\end{equation}
At this point the Hubbard--Stra\-to\-no\-vich transformation (\ref{Hub_C}) is applied. It yields for $p=s=1$ or $p=s=2$
\begin{align}
	\notag
	\exp \Bigg\{& \cumul{p}{s} \Bigg(\sum\limits_{k = 1}^n (-1)^p\prod\limits_{r = s}^{1} \aSclC_{k\beta_{r}}^{\phantom{+}} \prod\limits_{q = p}^{1} \aConjScl{\aSclB_{k\alpha_{q}}}\Bigg) \Bigg(\sum\limits_{k = 1}^n \prod\limits_{q = 1}^{p} \aSclB_{k\alpha_{q}} \prod\limits_{r = 1}^{s} \aConjScl{\aSclC_{k\beta_{r}}}\Bigg) \Bigg\} \\
	\notag
	&= \frac{n}{\pi} \int \exp \Bigg\{
	- \sum\limits_{k = 1}^n \tilde{\sclA}_{\beta\alpha}^{(k,p,s)} \cSclA_{\alpha\beta}^{(p,s)} - \sum\limits_{k = 1}^n \cConjScl{\cSclA}_{\alpha\beta}^{(p,s)} \sclA_{\alpha\beta}^{(k,p,s)} - n\abs{\cSclA_{\alpha\beta}^{(p,s)}}^2 \Bigg\} \\
	\label{HS commuting}
	&\times d\cConjScl{\cSclA}_{\alpha\beta}^{(p,s)}d\cSclA_{\alpha\beta}^{(p,s)},
\end{align}
where
\begin{equation}\label{chi def}
\begin{split}
\tilde{\sclA}_{\beta\alpha}^{(k,p,s)} &= \sqrt{n\cumul{p}{s}}(-1)^p\prod\limits_{r = s}^{1} \aSclC_{k\beta_{r}}^{\phantom{+}} \prod\limits_{q = p}^{1} \aConjScl{\aSclB_{k\alpha_{q}}}; \\
\sclA_{\alpha\beta}^{(k,p,s)}&= \sqrt{n\cumul{p}{s}} \prod\limits_{q = 1}^{p} \aSclB_{k\alpha_{q}} \prod\limits_{r = 1}^{s} \aConjScl{\aSclC_{k\beta_{r}}}.
\end{split}
\end{equation}

Then the combination of \eqref{IR 2}, \eqref{sum order change} and \eqref{HS commuting} 
gives us
\begin{align}
	\CF_2(z_1,z_2) = \left(\frac{n}{\pi}\right)^{5} \int&   e^{-\tr Q^*Q-|v|^2} \prod\limits_{k = 1}^n \mathsf{j}_k \,dQ^* dQ dv^* dv
	\label{after HS}
\end{align}
where 
\begin{align}
	\label{j_k def}
	\begin{split}
		\mathsf{j}_k &= \int \exp \left\{ \cSclB_{k,2} + \cSclB_{k,4} 
        \right\} d\aConjMat{\aVecBt_k} d\aVecBt_k d\aConjMat{\aVecCt_k} d\aVecCt_k,
	\end{split} \\
	\notag
	\cSclB_{k,2} &= - \left(\tr \tilde{\matA}_{k,1} Q + \tr Q^* \matA_{k,1}\right) + \aConjMat{\aVecBt_k}Z\aVecBt_k + \aConjMat{\aVecCt_k}Z^*\aVecCt_k, \\
	\label{S_k,4 def}
	\cSclB_{k,4} &= - \left( v\tilde{\matA}_{k,2}  +  \bar v \matA_{k,2}\right), 
\end{align}
In the formulas above 
$\cMatA$, $\tilde{\matA}_{k,1}$ and $\matA_{k,1}$ are matrices whose entries are 
$\cSclA_{\alpha\beta}^{(1,1)}$, $\tilde{\sclA}_{\beta\alpha}^{(k,1,1)}$ and $\sclA_{\alpha\beta}^{(k,1,1)}$ with $\alpha,\beta=1,2$ respectively, and $v=\cSclA_{\alpha\alpha}^{(2,2)}$,
$\tilde{\matA}_{k,2}=\tilde{\cSclA}_{\alpha\alpha}^{(k,2,2)}$, $\matA_{k,2}=\sclA_{\alpha\alpha}^{(k,2,2)}$ with $\alpha=\{1,2\}$. Therefore, $v\in \mathbb{C}$, and $Q$ is $2\times 2$ complex matrix.


Fortunately, the integral in \eqref{after HS} over $\aMatB$ and $\aMatC$ factorizes. Therefore the integration can be performed over $\aVecBt_k$ and $\aVecCt_k$ separately for every $k$. Lemma \ref{lem:int over psi,tau} provides a corresponding result.
\begin{lem}\label{lem:int over psi,tau}
	Let $\mathsf{j}_k$ be defined by \eqref{j_k def}. Then
	\begin{equation}\label{j_k value}
	\mathsf{j}_k = \det A + bv\det Q^* + b\bar{v}\det Q + b^2 \abs{v}^2,
	\end{equation}
	where $A$ is defined in \eqref{A def} and $b = \sqrt{n\cumul{2}{2}} = \sqrt{\frac{2(n - p)}{pn}}$.
\end{lem}
\begin{proof}
	The integral $\mathsf{j}_k$ is computed by the expansion of the exponent $e^{\cSclB_{k,4}}$ into series. We have
	\begin{equation} \label{psi,theta int}
		\begin{split}
			\mathsf{j}_k ={} &\int \left(1 + \cSclB_{k,4} + \frac{\cSclB_{k,4}^2}{2}\right)e^{ \cSclB_{k,2} } d\aConjMat{\aVecBt_k} d\aVecBt_k d\aConjMat{\aVecCt_k} d\aVecCt_k,
		\end{split}
	\end{equation}
	because $a_{k,4}^3 = 0$. Recalling the definition \eqref{chi def} of $\sclA_{\alpha\beta}^{(k,p,s)}$ and the values \eqref{cumul values} of $\cumul{p}{s}$, one can render $\cSclB_{k,2}$ and $\cSclB_{k,4}$ in the form
	\begin{align}\label{new form of b_2}
	\cSclB_{k,2} &= - \aConjMat{\aVecD_k}A\aVecD_k, \\
    \cSclB_{k,4} &= - b\left(\aSclC_{k2}^{\phantom{+}} \aSclC_{k1}^{\phantom{+}} \aConjScl{\aSclB_{k2}}\aConjScl{\aSclB_{k1}} v+ \bar v \aSclB_{k1}\aSclB_{k2} \aConjScl{\aSclC_{k1}}\aConjScl{\aSclC_{k2}}\right),
	\end{align}
	where $A$ is defined in \eqref{A def}, \begin{equation}\label{Psi def}
    \aVecD_k = \begin{pmatrix}
    \aVecBt_k \\
    \aVecCt_k
    \end{pmatrix}
    \end{equation}
    and $b = \sqrt{n\cumul{2}{2}} = \sqrt{\frac{2(n - p)}{pn}}$. Then we can compute the integral \eqref{psi,theta int} term-wise using \eqref{Grass int}:
    \begin{align*}
        &\int e^{ - \aConjMat{\aVecD_k}A\aVecD_k } d\aConjMat{\aVecD_k} d\aVecD_k = \det A; \\
        &- \int b\left(\aSclC_{k2}^{\phantom{+}} \aSclC_{k1}^{\phantom{+}} \aConjScl{\aSclB_{k2}}\aConjScl{\aSclB_{k1}} v + \bar v \aSclB_{k1}\aSclB_{k2} \aConjScl{\aSclC_{k1}}\aConjScl{\aSclC_{k2}}\right) e^{ - \aConjMat{\aVecD_k}A\aVecD_k } d\aConjMat{\aVecD_k} d\aVecD_k = b(v\det Q^* + \bar v\det Q); \\
        &\frac{b^2}{2}\int \left(\aSclC_{k2}^{\phantom{+}} \aSclC_{k1}^{\phantom{+}} \aConjScl{\aSclB_{k2}}\aConjScl{\aSclB_{k1}} v  + \bar v\aSclB_{k1}\aSclB_{k2} \aConjScl{\aSclC_{k1}}\aConjScl{\aSclC_{k2}}\right)^2 e^{ - \aConjMat{\aVecD_k}A\aVecD_k } d\aConjMat{\aVecD_k} d\aVecD_k = \\
        &\hspace{5 cm}  b^2 \int \abs{v}^2 \aSclC_{k2}^{\phantom{+}} \aSclC_{k1}^{\phantom{+}} \aConjScl{\aSclB_{k2}}\aConjScl{\aSclB_{k1}} \aSclB_{k1}\aSclB_{k2} \aConjScl{\aSclC_{k1}}\aConjScl{\aSclC_{k2}} d\aConjMat{\aVecD_k} d\aVecD_k = b^2 \abs{v}^2.
    \end{align*}
\end{proof}
A substitution of \eqref{j_k value} into \eqref{after HS}  gives \eqref{IR result}.
\section{Saddle-point analysis}
In order to perform the asymptotic analysis of (\ref{IR result}) let us change the variables $$Q = UT V^*, \quad v \to v \det UV^*,$$ 
where $Q = UT V^*$ is the singular value decomposition of the matrix $Q$, i.e. $$T = \diag\{t_j\}_{j = 1}^2, \quad t_j \ge 0, \quad U, V \in U(2).$$  The  Jacobian  of such change is $
    2\pi^{4}\Vanddet^2(T^2) \prod\limits_{j = 1}^2 t_j$ (see e.g.\ \cite{Hu:63}), and hence we obtain
	\begin{equation}\label{IR SVD}
	\begin{split}
	\CF_2 (z_1,z_2) = \frac{2n^5}{\pi} &\int\limits_\idom \Vanddet^2(T^2) t_1t_2 \exp\left\{nf(T, U, V, v)\right\} \\
	&\times d\mu(U) d\mu(V) dT d\bar v d v,
	\end{split}
	\end{equation}
	where $$\idom = \{(T, U, V, v) \mid t_j \ge 0,\, j = 1, 2,\, U, V \in U(2)\}, \quad dT = \prod\limits_{j = 1}^2 dt_j,$$ $\mu$ is a Haar measure on $U(2)$, and
	\begin{align}
        \label{f def}
        f(T, U, V, v) &= f_0(T, v) + \frac{1}{\sqrt{n}}f_r(T, U, V, v); \\
		\label{f_0 def}
		f_0(T, v) &= - \tr T^2 - |v|^2 + \log h_0(T, v); \\
		\label{h_0 def}
		h_0(T, v) &= \prod\limits_{j = 1}^2 (\abs{z_0}^2 + t_j^2) + 2bt_1t_2 \Re v + b^2 |v|^2; \\
		\label{f_r def}
		f_r(T, U, V, v) &= \sqrt{n}(\hat{f}(UT V^*, v \det UV^*) - f_0(T, v))
	\end{align}
 with $\hat f$ of (\ref{tilde f def}). Notice that $h$ of (\ref{h def}) in new variables takes the form
 \begin{equation}\label{h_v}
 h(Q,v) = \det A + bv\det T^* + b\bar v\det T + b^2 \abs{v}^2. 
 \end{equation}
\subsection{Saddle-points}
 
Since we use the Laplace method to analyse (\ref{IR SVD}), we are interested in the saddle-points where the global maximum of $f_0(T,v)=f_0(T,x,y)$ with $v=x+iy$
is achieved.

 We start with the following simple lemma
\begin{lem}\label{lem:max of f_0}
	The function $f_0 \colon [0, \infty)^{2} \times \Compl \to \R$ defined by \eqref{f_0 def} attains its global maximum value. Moreover, if $(\hat{T}, \hat{v})$ is a point of the global maximum then $\hat{t}_1 = \hat{t}_2$.
\end{lem}
\begin{proof}
    The function $f_0(T, v)$ is continuous and
    \begin{equation*}
        \lim_{\substack{T \to +\infty \\ v \to \infty}} f_0(T, v) = -\infty.
    \end{equation*}
    Therefore, $f_0$ attains its global maximum. Next, by AM-GM and QM-GM inequalities
    \begin{align}
        \label{AMGM}
        (\abs{z_0}^2 + t_1^2)(\abs{z_0}^2 + t_2^2) &\le \left(\abs{z_0}^2 + \frac{t_1^2 + t_2^2}{2}\right)^2, \\
        \notag
        2bt_1t_2x &\le b(t_1^2 + t_2^2)\abs{x}.
    \end{align}
    Note that the inequality \eqref{AMGM} is strict if $t_1 \ne t_2 \ge 0$. Hence, for $t_1 \ne t_2$ we have
    \begin{equation*}
    \begin{split}
        f_0(T, v) = - \tr T^2 - \abs{v}^2 + \log \left[ (\abs{z_0}^2 + t_1^2)(\abs{z_0}^2 + t_2^2) + 2bt_1t_2 x + b^2 \abs{v}^2 \right] \\
        < - \tr T^2 - \abs{v}^2 + \log \left[ \left(\abs{z_0}^2 + \frac{t_1^2 + t_2^2}{2}\right)^2 + b(t_1^2 + t_2^2)\abs{x} + b^2 \abs{v}^2 \right] = f_0(tI, \abs{x} + iy),
    \end{split}
    \end{equation*}
    where $t = \sqrt{\frac{t_1^2 + t_2^2}{2}}$. Thus, a point with $t_1 \ne t_2$ can not be a global maximum point of $f_0$.
\end{proof}
Now we are ready to prove
\begin{prop}\label{p:main_saddle}
    Function $f_0(T,v)=f_0(t_1,t_2,x,y)$ with $v=x+iy$ of (\ref{f_0 def}) may attain its global maximum only at the saddle-points of these three types:
\begin{enumerate}
    \item[I.] $*$-saddle point: $t_1=t_2=t_*$, $x=x_*$, $y=0$.
    
Here $x_*=\alpha b$ and $\alpha\in [0,1]$ is a solution of
\begin{align}\label{alpha}
    2\alpha (1-\alpha)^2b^2+1-\alpha=|z_0|^2,
\end{align}
such that
\begin{equation}\label{der1}
     (6\alpha^2-8\alpha+2)b^2-1\le  0,
   \end{equation}
and
\begin{align}\label{t*}
    &t_*^2=\dfrac{\alpha |z_0|^2}{1-\alpha}-\alpha b^2\ge 0.
\end{align}
The value of $f_0$ at this point is
    \begin{equation}\label{v_I}
      F_{I}(\alpha, b,|z_0|^2)=f_0(t_1,t_2,x,y)\Big|_{I}=-\alpha^2 b^2-\dfrac{2\alpha |z_0|^2}{1-\alpha}+2\alpha b^2+\log \dfrac{|z_0|^2}{1-\alpha}.
    \end{equation}
    
    \item[II.] $v$-saddle points: $t_1=t_2=0$, $|v|^2=x^2+y^2=1-|z_0|^4/b^2$

    The value of $f_0$ at this point is
    \begin{equation}\label{v_II}
      F_{II}(b,|z_0|^2)=f_0(t_1,t_2,x,y)\Big|_{II}=-(1-\dfrac{|z_0|^4}{b^2})+\log b^2.  
    \end{equation}
    \item[III.] Zero saddle-point: $t_1=t_2=x=y=0$. 

    The value of $f_0$ at this point is
    \begin{equation}\label{v_III}
     F_{III}(b,|z_0|^2)= f_0(t_1,t_2,x,y)\Big|_{III}=\log |z_0|^4.  
    \end{equation}
\end{enumerate}
\end{prop}

\begin{proof}
    
Taking the derivatives of $f_0$ of (\ref{f_0 def}), we get the following system of equations for the stationary points
\begin{align*}
\begin{cases}
    -x+\dfrac{bt_1t_2+b^2x}{h_0(T,v)}=0\\
    -y+\dfrac{b^2y}{h_0(T,v)}=0\\
    -t_1+\dfrac{t_1(|z_0|^2+t_2^2)+bxt_2}{h_0(T,v)}=0\\
    -t_2+\dfrac{t_2(|z_0|^2+t_1^2)+bxt_1}{h_0(T,v)}=0
\end{cases}
\end{align*}
According to Lemma \ref{lem:max of f_0}, we want to consider only points with $t_1=t_2=t$, so the system transforms to 
\begin{align}\label{system}
\begin{cases}
    -x+\dfrac{bt^2+b^2x}{h_0(t,v)}=0\\
    -y+\dfrac{b^2y}{h_0(t,v)}=0\\
    -t+\dfrac{t(|z_0|^2+t^2)+bxt}{h_0(t,v)}=0\\
\end{cases}
\end{align}
To find the solutions of (\ref{system}), consider first the case $y=0$. 
 Then (\ref{system}) can be rewritten as
\begin{align*}
\begin{cases}
    -x+\dfrac{bt^2+b^2x}{h_0(T,v)}=0\\
    -t+\dfrac{t(|z_0|^2+t^2)+bxt}{h_0(T,v)}=0
\end{cases}
\end{align*}
Clearly, if $t=0$, then we can have $x=0$, or
\begin{equation*}
h_0(0,v)=b^2\Longrightarrow x^2=1-\dfrac{|z_0|^4}{b^2}.
\end{equation*}
which gives two solutions:
\begin{align}\label{sol2}
    &t_1=t_2=t=0,\quad x=y=0,\\ \label{sol4_1}
    &t_1=t_2=t=0,\quad x^2=1-\dfrac{|z_0|^4}{b^2}, \quad y=0.
\end{align}
If $t\ne 0$, we get
\begin{equation}\label{h_*1}
h_0(T,v)=bx+t^2+|z_0|^2,
\end{equation}
and hence from the first equation of (\ref{system})
\begin{equation}\label{t_*}
bt^2+b^2x=xh_0(T,v)\Longrightarrow t^2=\dfrac{x|z_0|^2}{b-x}-bx
\end{equation}
if $x\ne b$. Substituting this to the previous equation we get
\[
bx+t^2+|z_0|^2=\dfrac{b|z_0|^2}{b-x},
\]
and hence
\begin{equation}\label{B_*}
h_0(T,v)=\dfrac{b^2|z_0|^4}{(b-x)^2}-2bx|z_0|^2=\dfrac{b|z_0|^2}{b-x}\Longrightarrow 2x(b-x)^2+b-x=|z_0|^2b,
\end{equation}
if $b\ne 0$ and $|z_0|\ne 0$. The last one is a cubic equation with respect to $x$, which may have three real roots. 
It is convenient to write
\begin{equation}\label{alp}
 x=\alpha b,
\end{equation}
and hence from (\ref{t_*}), (\ref{B_*}) we get
\begin{align}\label{alp_expr}
  &t^2=  \dfrac{\alpha |z_0|^2}{1-\alpha}-\alpha b^2,\quad h_0(T,v)=\dfrac{ |z_0|^2}{1-\alpha},\\ \notag
  &2\alpha(1-\alpha)^2b^2+1-\alpha=|z_0|^2.
\end{align}
Now we need
\begin{lem}\label{l:*}
    Among all stationary points corresponding to the real roots of
    the equation
    \begin{equation}\label{cub_eq}
    2\alpha(1-\alpha)^2b^2+1-\alpha=|z_0|^2,
    \end{equation}
    the global maximum of $f_0$ can be achieved only at the saddle point
    corresponding to the solution $\alpha$ of (\ref{cub_eq}) lying on $[0,1]$ and such that
   \begin{equation}\label{der}
     (6\alpha^2-8\alpha+2)b^2-1\le  0. 
   \end{equation}
   The saddle-point corresponding to this solution exists if and
   only if  one of the following holds
   \begin{enumerate}
       \item[(i)] $b\le \tfrac{1}{\sqrt2}$, $|z_0|^2\le 1$;
       \item[(ii)] $b\in (\tfrac{1}{\sqrt2}, 1]$, $|z_0|^2\le z_-(b),$ where
       \begin{equation}\label{3root}
    z_-(b)=\dfrac{4b^2+9}{27}+\dfrac{4b^2+6}{27}\sqrt{1+\tfrac{3}{2b^2}};
    \end{equation} 
    \item[(iii)] $b\in (1,\sqrt{\tfrac{5+\sqrt5}{5}})$, $|z_0|^2\in [\tfrac{b^2-b\sqrt{2-b^2}}{2}, z_-(b)]$; 
    \item [(iv)] $b\in [\sqrt{\tfrac{5+\sqrt5}{5}},\sqrt 2]$, $|z_0|^2\in [\tfrac{b^2-b\sqrt{2-b^2}}{2}, \tfrac{b^2+b\sqrt{2-b^2}}{2}]$;
   \end{enumerate}
   In addition, at this saddle-point 
    \begin{equation}\label{b,alp_in}
    b^2\le \dfrac{1}{1-2\alpha(1-2\alpha)},
    \end{equation}
    \begin{equation}\label{h_*}
      h_*=h_0(t_*,x_*,0)=\dfrac{|z_0|^2}{1-\alpha}=bx_*+t_*^2+|z_0|^2.  
    \end{equation}
\end{lem}
\begin{proof}
Notice that as soon as (\ref{cub_eq}) is satisfied
we get
\begin{align*}
    f_0(T,v)&=-2t^2-x^2+\log h_0(T,v)=\alpha^2b^2-\dfrac{2\alpha|z_0|^2}{1-\alpha}+2b^2\alpha(1-\alpha)+\log\dfrac{|z_0|^2}{1-\alpha}\\
    &=\alpha^2b^2-\dfrac{|z_0|^2}{1-\alpha}+\log \dfrac{|z_0|^2}{1-\alpha}+2|z_0|^2-1.
\end{align*}
Here we used (\ref{alp}) -- (\ref{alp_expr}) and
\[
2b^2\alpha(1-\alpha)=\dfrac{|z_0|^2}{1-\alpha}-1.
\]
Taking the derivative with respect to $\alpha$ of the function above we get
\begin{equation}\label{der_f_0}
2\alpha b^2-\dfrac{|z_0|^2}{(1-\alpha)^2}+\dfrac{1}{1-\alpha}=\dfrac{2\alpha(1-\alpha)^2b^2+1-\alpha-|z_0|^2}{(1-\alpha)^2}.
\end{equation}
Suppose (\ref{cub_eq}) has three real roots. Then, according to (\ref{der_f_0}), two of them are the local minimums, and the middle one is a local maximum of the function above, and hence
the value of $f_0$ in the saddle-point associated to the middle root of (\ref{cub_eq}) is greater then those of other two.
Notice that if $p(\alpha)=2\alpha(1-\alpha)^2b^2+1-\alpha-|z_0|^2$, then
\[
p(0)=1-|z_0|^2,\quad p(1)=-|z_0|^2,
\]
and hence the biggest root of $p(\alpha)$ is always grater than $1$.
If $|z_0|^2\le 1$, then (\ref{cub_eq}) must has a non-positive root and the root between $0$ and $1$, and so the middle root belongs to $[0,1]$. If $|z_0|^2>1$, then one root is still bigger $1$, and two other roots, if exist, must be of the same sign (since the product of three roots is $(|z_0|^2-1)/2b^2$).  
It is easy to see that the point with negative $x$ and $t\ne 0$ cannot be a global maximum of $f_0$ since changing the sign of $x$ evidently increases the value of $f_0$, thus we interested only in the case when all three roots are positive. 
Notice that both of them should be at the same side of $1$, and since the sum of all three roots is $2$ according to Vieta's formula, it implies that both of
them must lie on $[0,1]$. In addition, at the middle root $p'$ must be negative which gives (\ref{der1}).

It remains to consider the case when $p(\alpha)$ has only one real root. But in this case no saddle-point corresponding to the root exists since we get $t^2<0$ (see (\ref{alp_expr})).

Thus, the saddle-point corresponding to the solution of (\ref{cub_eq}) exists and can be a global maximum of $f_0$ only if (\ref{cub_eq}) has a solution $\alpha\in[0,1]$ and this solution satisfies (\ref{der1}) and
\begin{align}\label{t_cond}
  \dfrac{\alpha|z_0|^2}{1-\alpha}\ge \alpha b^2\Longrightarrow 2\alpha(1-\alpha)b^2+1\ge b^2 \Longleftrightarrow \alpha\in [\tfrac{b-\sqrt{2-b^2}}{2b},\tfrac{b+\sqrt{2-b^2}}{2b}].
\end{align}
The last condition guarantees the existence of $t_*$ (see (\ref{t*})) and implies (\ref{b,alp_in}).
Notice that since $0\le 2\alpha(1-\alpha)\le 1/2$, we immediately get that $b\le \sqrt 2$ and (\ref{t_cond}) is always satisfied if $b\le 1$.

 From the discussion above, if $|z_0|\le 1$, then the solution $\alpha$ of (\ref{cub_eq}) satisfying (\ref{der1}) exists. Notice that
\begin{equation}\label{alp_pm}
p'(\alpha)=0\Longleftrightarrow \alpha=\alpha_\pm=\dfrac{2\pm \sqrt{1+\tfrac{3}{2b^2}}}{3}.
\end{equation}
Therefore, since $\alpha_+>1$, if $|z_0|>1$, then two positive solutions on $[0,1]$ exist if 
\[
\alpha_->0,\quad p(\alpha_-)\ge 0
\]
which is equivalent to $b>\tfrac{1}{\sqrt 2}$ and $|z_0|^2\le z_-(b)=p(\alpha_-) + \abs{z_0}^2$ which gives $z_-(b)$ of (\ref{3root}). 
Therefore, if $b\le 1$, then $*$-saddle point exists if and only if (i) or (ii) holds,
and if $b\in (1,\sqrt2]$ and $|z_0|^2\le z_-(b)$, then the solution $\alpha$ of (\ref{cub_eq}) satisfying (\ref{der1}) exists.

If $b\in (1,\sqrt2]$, then in order to satisfy (\ref{t_cond}) we must have 
\begin{equation}\label{al_int}
  \alpha\in [\tfrac{b-\sqrt{2-b^2}}{2b},\tfrac{b+\sqrt{2-b^2}}{2b}].  
\end{equation}
Notice that if $b\in [\sqrt{\tfrac{5+\sqrt5}{5}},\sqrt2]$, then 
\[
[\tfrac{b-\sqrt{2-b^2}}{2b},\tfrac{b+\sqrt{2-b^2}}{2b}]\subset [\alpha_-,1],
\]
and hence $p'(\alpha)<0$ for all $\alpha$ satisfying (\ref{al_int}). Therefore, if
\[
|z_0|^2\in [p(\tfrac{b+\sqrt{2-b^2}}{2b}),p(\tfrac{b-\sqrt{2-b^2}}{2b})]=[\tfrac{b^2-b\sqrt{2-b^2}}{2}, \tfrac{b^2+b\sqrt{2-b^2}}{2}],
\]
then the solution $\alpha_*$ of (\ref{cub_eq}) satisfying (\ref{der1}) and (\ref{al_int}) exists, and so does $*$-saddle point, which gives (iv). 

If $b\in [1,\sqrt{\tfrac{5+\sqrt5}{5}})$, then
\[
\tfrac{b-\sqrt{2-b^2}}{2b}<\alpha_-<\tfrac{b+\sqrt{2-b^2}}{2b}<1,
\]
and so possible values of $|z_0|^2$ should correspond to $\alpha\in [\alpha_-,\tfrac{b+\sqrt{2-b^2}}{2b}]$ which gives (iii).

The expression for $h_*$ can be obtained straightforwardly from (\ref{alp_expr}).
\end{proof}
If $b=0$, then the first equation of (\ref{system}) implies $x=0$, and then
\[
h_0(T,v)=t^2+|z_0|^2 \Longrightarrow (t^2+|z_0|^2)^2=t^2+|z_0|^2\Longrightarrow t=\sqrt{1-|z_0|^2}
\]
for $|z_0|\le 1$. 

If $|z_0|=0$, then
\[
h_0(T,v)=t^2+bx\Longrightarrow (t^2+bx)^2=t^2+bx\Longrightarrow t^2+bx=1,
\]
and hence the first equation of (\ref{system}) gives $x=b$, thus
\[
t_1=t_2=\sqrt{1-b^2},\quad x=b,\quad y=0
\]
for $b\le 1$. Notice also that $x=b\ne 0$ always implies 
\[
t^2+b^2=t^2+b^2+|z_0|^2\Longrightarrow |z_0|=0.
\]
Thus, the case $t_1=t_2=t\ne 0$, $y=0$ for $|z_0|\ne 0$, $b\ne 0$ gives the solution
\begin{align}\label{sol3}
t_1=t_2=t_*,\quad x=x_*=\alpha b,\quad y=0    
\end{align}
where $\alpha\in [0,1]$ and 
\[
2\alpha(1-\alpha)^2b^2+1-\alpha=|z_0|^2,\quad t_*^2=\dfrac{\alpha |z_0|^2}{1-\alpha}-\alpha b^2
\]
if such solution exists and (\ref{b,alp_in}) holds (see Lemma \ref{l:*}).

If $z_0=0$, but $b\le 1$, then the solution takes the form
\[
t_1=t_2=\sqrt{1-b^2},\quad x=b,\quad y=0.
\]
Notice that the last solution can be considered as a limiting case of (\ref{sol3}), so one can consider (\ref{sol3}) also for $|z_0|=0$.
Same is true for the solution
\[
t_1=t_2=\sqrt{1-|z_0|^2},\quad x=y=0
\]
obtained for $b=0$, $|z_0|\le 1$.

Assume now that $y\ne 0$. Then the second equation of (\ref{system})
gives 
\[
h_0(T,v)=b^2,
\]
and hence the first equation gives $t=0$.
Thus the previous equation implies
\[
|z_0|^4+b^2(x^2+y^2)=b^2\Longrightarrow x^2+y^2=1-\dfrac{|z_0|^4}{b^2}
\]
if $|z_0|^2\le b$.
Hence, another family of solutions takes the form
\begin{equation}\label{sol4}
 t_1=t_2=0,\quad x^2+y^2=1-\dfrac{|z_0|^4}{b^2}   
\end{equation}
if $|z_0|^2\le b$. Notice that this includes (\ref{sol4_1}), so one can include here $y=0$. The equation (\ref{h_*}) follows from (\ref{h_*1}) and (\ref{alp_expr}).

\end{proof}

\subsection{Main saddle-points}
Depending on the values of $|z_0|$ and $b$, the main contribution to (\ref{IR SVD}) is given by the different saddle-points.

Notice first that if $|z_0|^2\le b$ and so the $v$-saddle point exists, then 
\begin{equation}\label{II,III_in}
    F_{II}(b,|z_0|^2)\ge F_{III}(b,|z_0|^2).
\end{equation}

Indeed, according to (\ref{v_II}) -- (\ref{v_III}) 
\[
F_{II}(b,|z_0|^2)- F_{III}(b,|z_0|^2)=\dfrac{|z_0|^4}{b^2}-1-\log \dfrac{|z_0|^4}{b^2}\ge 0. 
\]
Compare now the values at $*$-saddle point and $v$-saddle point:
\begin{lem}\label{l:I>II}
Let $ |z_0|^2\le b$ (i.e. $v$-saddle point exists).
Then
\begin{enumerate}
    \item[(i)] If one of the following holds
    \begin{itemize}
        \item $b\le 1$, $|z_0|\le 1$
        \item $b\in (1, \sqrt2]$ and
    \[  \dfrac{b^2-b\sqrt{2-b^2}}{2}\le |z_0|^2\le 1,\]
    \end{itemize}
    then $*$-saddle point exists and 
\[
F_{I}(\alpha, b,|z_0|^2)\ge F_{II}(b,|z_0|^2),
\]
where $F_I, F_{II}$ are defined in (\ref{v_I}) -- (\ref{v_II}).
\item[(ii)] If $b\in (1, \sqrt2]$, but $|z_0|> 1$, then 
there exists a curve $z_1(b)$ such that $*$-saddle point exists for $1\le |z_0|^2\le \min(z_1(b),b)$ and
\begin{align}\label{I>II}
F_{I}(\alpha, b,|z_0|^2)\ge F_{II}(b,|z_0|^2),\quad 1\le |z_0|^2\le \min(z_1(b),b).
\end{align}
If $ \min(z_1(b),b)\le |z_0|^2\le b$, then or $*$-point does not exist, or 
\begin{align}\label{I<II}
&F_{I}(\alpha, b,|z_0|^2)< F_{II}(b,|z_0|^2).
\end{align}

\end{enumerate}
Therefore, the curve $\gamma_2$ on Figure 1 is 
\[
|z_0|^2=z_1(b), \quad b: z_1(b)>b.
\]
It coincides with 
\[
|z_0|^2=\dfrac{b^2+b\sqrt{2-b^2}}{2}
\]
for all $b\in [\sqrt{\tfrac{5+\sqrt 5}{5}};\sqrt 2]$.
\end{lem}
\begin{proof}
We start with (i). Since $z_-(b)\ge 1$ (see (\ref{3root})) and
\[
\frac{b^2+b\sqrt{2-b^2}}{2}\ge 1
\]
for $b\ge 1$ , the existence of $*$-saddle point in the conditions of (i) follows from Lemma \ref{l:*}.

Now, according to (\ref{alpha}), we have
\[
\dfrac{|z_0|^2}{1-\alpha}=2\alpha(1-\alpha)b^2+1.
\]
Therefore,
\begin{align*}
&F_{I}(\alpha, b,|z_0|^2)- F_{II}(b,|z_0|^2)
=-|z_0|^2(\tfrac{|z_0|^2}{b^2}+\tfrac{2\alpha}{1-\alpha})+2\alpha b^2-\alpha^2 b^2+1\\ 
&+\log(2\alpha(1-\alpha)+\tfrac{1}{b^2})=
b^2(2\alpha-5\alpha^2+4\alpha^3-4\alpha^2(1-\alpha)^4)-\dfrac{(1-\alpha)^2}{b^2}\\ &+(1-2\alpha-4\alpha(1-\alpha)^3)+\log(2\alpha(1-\alpha)+\tfrac{1}{b^2}),\notag
\end{align*}
where to obtain the last equality we substitute $|z_0|^2=2\alpha (1-\alpha)^2b^2+1-\alpha$ and open the parentheses.

Fix $\alpha$ and consider $s=1/b^2\ge 1-2\alpha(1-\alpha)$ (recall that if $*$-saddle point exists, then we have (\ref{b,alp_in})). Denote
\begin{align}\label{H}
H(s)=&s^{-1}\cdot (2\alpha-5\alpha^2+4\alpha^3-4\alpha^2(1-\alpha)^4)-(1-\alpha)^2 s\\ \notag
&+(1-2\alpha-4\alpha(1-\alpha)^3)+\log(2\alpha(1-\alpha)+s).    
\end{align}
One can easily check that
\begin{equation}\label{H=0}
H(1-2\alpha(1-\alpha))=0.
\end{equation}
Moreover, computing
\[
H'(s)=-\dfrac{2\alpha-5\alpha^2+4\alpha^3-4\alpha^2(1-\alpha)^4}{s^2}-(1-\alpha)^2+\dfrac{1}{2\alpha(1-\alpha)+s},
\]
one can also check
\[
H'(1-2\alpha(1-\alpha))=0
\]
and
\begin{equation*}
H'(s)=\dfrac{s-1+2\alpha(1-\alpha)}{s^2(s+2\alpha(1-\alpha))}\cdot g(s),  
\end{equation*}
where
\begin{equation}\label{g}
 g(s)=-(1-\alpha)^2s^2+(2\alpha-\alpha^2)s+4\alpha^2(1-\alpha)^4+2\alpha^3(1-\alpha).
\end{equation}
Since
\[
\dfrac{s-1+2\alpha(1-\alpha)}{s^2(s+2\alpha(1-\alpha))}\ge 0
\]
for $s\ge 1-2\alpha(1-\alpha)$, the sign of $H'(s)$ is determined by the sign of $g(s)$.

Recall that we also have
\[
|z_0|^2=2\alpha (1-\alpha)^2b^2+1-\alpha\le b\Longrightarrow b\in (b_-(\alpha),b_+(\alpha)),
\]
where
\begin{equation}\label{b_pm}
b_{\pm}(\alpha)=\dfrac{1\pm\sqrt{1-8\alpha (1-\alpha)^3}}{4\alpha (1-\alpha)^2}.
\end{equation}
Since it is easy to check that $b_+(\alpha)>\sqrt 2$ and $b_-(\alpha)^2\le (1-2\alpha(1-\alpha))^{-1}$ for any $\alpha\in [0,1]$, we are interested in 
$b\in [b_-(\alpha),(1-2\alpha(1-\alpha))^{-1/2}]$, i.e.
\begin{equation}\label{s_0}
s\in [1-2\alpha(1-\alpha), s_0(\alpha)],\quad s_0(\alpha)=(b_-(\alpha))^{-2}.
\end{equation}
Notice that
\begin{align*}
&g(1-2\alpha(1-\alpha))=-(1-\alpha)^2(1-2\alpha(1-\alpha))^2+(2\alpha-\alpha^2)(1-2\alpha(1-\alpha))\\
&+4\alpha^2(1-\alpha)^4+2\alpha^3(1-\alpha)=-(4\alpha^2-6\alpha+1)(1-2\alpha(1-\alpha)),
\end{align*}
and hence we get
\begin{align}\label{sign_g}
&g(1-2\alpha(1-\alpha))< 0, \quad \alpha\in [0,\tfrac{3-\sqrt5}{4}),\\ \notag
& g(1-2\alpha(1-\alpha))\ge 0, \quad \alpha\in [\tfrac{3-\sqrt5}{4},1]
\end{align}

Now we need the following simple lemma
\begin{lem}\label{l:sqrt2}
    If $|z_0|^2\le 1$, $b\le (1-2\alpha(1-\alpha))^{-1/2}$ and $|z_0|^2\le b$, then $\alpha\ge 1-\tfrac{1}{\sqrt2}$.
\end{lem}
\begin{proof}
 Indeed, according to (\ref{alpha}), $|z_0|^2\le 1$ implies
 \begin{equation}\label{al_in}
 2\alpha(1-\alpha)^2b^2+1-\alpha\le 1\Longrightarrow \alpha\ge 1-\dfrac{1}{b\sqrt2},
 \end{equation}
 which gives the statement for $b\ge 1$. If $b\le 1$, then we must have $b_-(\alpha)\le 1$ (see (\ref{b_pm})) which also implies $1-\alpha\le 1/\sqrt2$.
\end{proof}

 According to Lemma \ref{l:sqrt2}, in the conditions of Lemma \ref{l:I>II} (i) we have $\alpha>1-1/\sqrt2>\tfrac{3-\sqrt5}{4}$, and
 thus $g(1-2\alpha(1-\alpha))\ge 0$. As $g$ is a quadratic polynomial with negative top coefficient, it has exactly one root on $[1-2\alpha(1-\alpha),\infty)$, hence, the same is true for $H'(s)$. Thus $H(s)$ increases for $[1-2\alpha(1-\alpha),a)$, and then decreases. 


Therefore, if $H(s_0(\alpha))\ge 0$, then $H(s)\ge 0$ for all $s\in [1-2\alpha(1-\alpha), s_0(\alpha)]$ and $\alpha\ge 1-1/\sqrt2$.

It remains to check 
\begin{equation}\label{H(s_0)}
H(s_0(\alpha))\ge 0, \quad \alpha\ge 1-1/\sqrt2.
\end{equation}
If $s=s_0(\alpha)$ (i.e. $b=b_-(\alpha)$), then 
\[
|z_0|^2=2\alpha(1-\alpha)^2b^2+1-\alpha=b,
\]
i.e.
\begin{align*}
 H(s_0(\alpha))&=-\dfrac{2\alpha b}{1-\alpha}+2\alpha b^2-\alpha^2b^2-\log(b(1-\alpha))\\
 &=b-1-\alpha^2b^2-\alpha+2\alpha^3b^2 -\log(b(1-\alpha))=:r(b,\alpha)  
\end{align*}
where we have used 
\begin{align*}
    &2\alpha b^2-\alpha^2b^2=2\alpha(1-\alpha)b^2+\alpha^2b^2=\dfrac{b}{1-\alpha}+\alpha^2b^2\\
    &\dfrac{b}{1-\alpha}=2\alpha(1-\alpha)b^2+1.
\end{align*}
Taking the derivative with respect to $\alpha$ we get
\[
r'_\alpha(b,\alpha)=\dfrac{\alpha}{1-\alpha}(1-b^2(2-8\alpha+6\alpha^2)).
\]
Taking into account (\ref{der1}), $r'_\alpha(b,\alpha)\ge 0$, and hence
$r(b,\alpha)$ increases in $\alpha$. At $\alpha=1-1/\sqrt2$ we get
\[
r(b,1-1/\sqrt2)=b-2+\dfrac{1}{\sqrt2}-(\sqrt2-1)^3b^2/2-\log \dfrac{b}{\sqrt2},
\]
which, as one can easily check, is positive for all $b$.
Therefore, $r(b,\alpha)\ge 0$ for all $\alpha\ge 1-1/\sqrt2$,
and hence we obtain (\ref{H(s_0)}), which finishes the proof of Lemma \ref{l:I>II} (i).

Let's now prove (ii). Suppose $|z_0|>1$. Proceeding similarly to the proof of (i), we want to study when $H(s)\ge 0$ for $ 1-2\alpha(1-\alpha)\le s\le s_0(\alpha)$,  where $H$ is defined at (\ref{H}) and $s_0(\alpha)$ is defined in (\ref{s_0}). Since $g(0)\ge 0$, $1-2\alpha(1-\alpha)>0$, and according to (\ref{sign_g}), we have for all $s> 1-2\alpha(1-\alpha)$
\[
g(s)<0,\quad \alpha< \dfrac{3-\sqrt5}4,
\]
and hence
\begin{align*}
    H'(s)< 0,\quad \alpha<\dfrac{3-\sqrt5}4\Longrightarrow H(s)<0\,\,\hbox{if}\,\, \alpha<\dfrac{3-\sqrt5}4, s> 1-2\alpha(1-\alpha).
\end{align*}
Therefore, if $\alpha<\tfrac{3-\sqrt5}4$, and $*$-saddle point exists,
then (\ref{I<II}) holds.

Let $\alpha\ge \tfrac{3-\sqrt5}4$. According to (\ref{sign_g}), then
\begin{align}\label{domain}
&H(s)\ge 0,\quad s\in [1-2\alpha(1-\alpha), s_1(\alpha)],\\ \notag
&H(s)< 0,\quad s\ge s_1(\alpha),
\end{align}
where $s_1(\alpha)$ is a solution of $H(\alpha, s)=0$ bigger than
$1-2\alpha(1-\alpha)$ (which exists since $H(s)\to -\infty$, $s\to\infty$) for $\alpha<1$. Notice that 
\[
s_1(\tfrac{3-\sqrt5}4)=1-2\alpha(1-\alpha)\Big|_{\alpha=\tfrac{3-\sqrt5}4}=\dfrac{5-\sqrt5}{4}.
\]

Recall that we are interested in $ 1-2\alpha(1-\alpha)\le s\le s_0(\alpha)$. If $F(s_0(\alpha))\ge 0$ (i.e. $s_1(\alpha)\ge s_0(\alpha)$), then $H(s)\ge 0$ for all such $s$. Numerically, one can compute that this happens if $\alpha\ge \alpha_0\approx 0.22$.

If $\tfrac{3-\sqrt5}4 \le \alpha\le \alpha_0$, then
$H(s)\ge 0$ for $s\in [1-2\alpha(1-\alpha), s_1(\alpha)]$, and
$H(s)<0$ for $s\in [s_1(\alpha), s_0(\alpha)]$. 

Notice also that $\tfrac{3-\sqrt5}4 \le \alpha\le 1$, then
\begin{equation}\label{der_s_1}
   6\alpha^2-8\alpha+2\le 1-2\alpha(1-\alpha)\le s\Longrightarrow  (6\alpha^2-8\alpha+2)b^2-1<0
\end{equation}
for all $b^2=1/s$ with $s\ge 1-2\alpha(1-\alpha)$.

Now we will need
\begin{lem}\label{l:s_1}
   In the notations above we have
   \[
   s'_1(\alpha)>0
   \]
 as soon as $\tfrac{3-\sqrt5}4 < \alpha\le \alpha_0$.
\end{lem}
\begin{proof}
Taking into account the definition of $s_1(\alpha)$, we get
\begin{align*}
    H(\alpha,s_1(\alpha))=0\Longrightarrow H'_\alpha(\alpha,s_1(\alpha))+s'_1(\alpha)\cdot H'_s(\alpha,s_1(\alpha))=0.
\end{align*}
Since $H'_s(\alpha,s_1(\alpha))<0$, it is enough to check that $H'_\alpha(\alpha,s_1(\alpha))>0$. Taking the derivative of (\ref{H}) one get
\begin{align*}
\dfrac{1}{2}H'_\alpha&=\dfrac{1}{s}\cdot (1-5\alpha+6\alpha^2-4\alpha(1-\alpha)^4+8\alpha^2(1-\alpha)^3)+(1-\alpha)s\\
&-1-2(1-\alpha)^3+6\alpha(1-\alpha)^2+\dfrac{1-2\alpha}{s+2\alpha(1-\alpha)}\\
&=\dfrac{s+2\alpha(1-\alpha)-1}{s(s+2\alpha(1-\alpha))}\cdot q(\alpha,s),
\end{align*}
where
\[
q(\alpha,s)=(1-\alpha)s^2+s(-\alpha-2\alpha^3+6\alpha(1-\alpha)^2)-2\alpha(1-\alpha)(1-3\alpha)(2(1-\alpha)^2-1).
\]
If $\tfrac{3-\sqrt5}4 \le \alpha\le \alpha_0\approx 0.22$, then
\[
q(\alpha,0)<0.
\]
In addition, 
\[
q(\alpha,1-2\alpha(1-\alpha))=8\alpha^3-16\alpha^2+8\alpha-1=(2\alpha-1)(4\alpha^2-6\alpha+1)>0 
\]
for $\tfrac{3-\sqrt5}4 < \alpha\le \alpha_0\approx 0.22$. Since $q(\alpha, s)$ is a quadratic polynomial in $s$ and $1-2\alpha(1-\alpha)>0$, the consideration above gives
\[
q(\alpha, s)> 0
\]
for all $s\ge 1-2\alpha(1-\alpha)$ including $s=s_1(\alpha)$, which implies the lemma.
\end{proof}
Notice that it is easy to check that $s_0'(\alpha)>0$ for $\alpha\ge \alpha_0\approx 0.22$. This and lemma implies that (\ref{domain}) can be rewritten as 
\begin{align*}
H(\alpha, s)\ge 0
\end{align*}
if
\begin{align}\label{domain_s}
&\tfrac{1}{2}\le s\le \tfrac{5-\sqrt5}{4},\quad \alpha\in [\tfrac{1-\sqrt{2s-1}}{2};\tfrac{1+\sqrt{2s-1}}{2}]\,\,\hbox{or}\\ \notag
& \tfrac{5-\sqrt5}{4}\le s\le 1, \quad \alpha\in [\tilde\alpha_1(s),\tfrac{1+\sqrt{2s-1}}{2}]
\end{align}
and $H(\alpha, s)< 0$ or $*$-saddle point does not exist in the remaining domain. Here $\tilde\alpha_1(s)$ is an inverse function to
$s_1(\alpha)$ for $\tfrac{3-\sqrt5}4 \le \alpha<\alpha_0$ and to
$s_0(\alpha)$ for $\alpha\in [\alpha_0,1]$. In terms of $b\in (1,\sqrt2]$ the domain (\ref{domain_s}) can be rewritten as
\begin{align}\label{domain_b}
& \sqrt{\tfrac{5+\sqrt5}{5}}\le b\le \sqrt2,\quad \alpha\in [\tfrac{b-\sqrt{2-b^2}}{2b};\tfrac{b+\sqrt{2-b^2}}{2b}]\,\,\hbox{or}\\ \notag
&1< b\le \sqrt{\tfrac{5+\sqrt5}{5}}, \quad \alpha\in [\alpha_1(b),\tfrac{b+\sqrt{2-b^2}}{2b}]
\end{align}
where $\alpha_1(b)=\tilde\alpha_1(1/b^2)$. 

It remains to notice that for $\alpha\ge \tfrac{3-\sqrt5}4$, $b^2\le \tfrac{1}{1-2\alpha(1-\alpha)}$
\[
(2\alpha(1-\alpha)^2b^2+1-\alpha)'_\alpha=(6\alpha^2-8\alpha+2)b^2-1<0,
\]
and hence in terms of $b$, $|z_0|^2$ (\ref{domain_b}) takes the form
\begin{align}\label{domain_z}
&\sqrt{\tfrac{5+\sqrt5}{5}}\le b\le \sqrt2,\quad |z_0|^2\in [\tfrac{b^2-b\sqrt{2-b^2}}{2};\tfrac{b^2+b\sqrt{2-b^2}}{2}]\,\,\hbox{or}\\ \notag
&1< b\le \sqrt{\tfrac{5+\sqrt5}{5}}, \quad |z_0|^2\in [\tfrac{b^2-b\sqrt{2-b^2}}{2}; z_1(b)]
\end{align}
with $$z_1(b)=2\alpha(1-\alpha)^2b^2+1-\alpha\Big|_{\alpha=\alpha_1(b)}.$$
According to the definition of $\alpha_0$, we get $z_1(b)=b$ for $1\le b\le 1/\sqrt{s_1(\alpha_0)}$.
Numerically, $1/\sqrt{s_1(\alpha_0)}\approx 1.11$. 

The existence of $*$-saddle point in the domain (\ref{domain_z}) follows from (\ref{der_s_1}) (which gives $z_1(b)\le z_-(b)$) and Lemma \ref{l:*}.
This finishes the proof of Lemma \ref{l:I>II}.
\end{proof}
\begin{lem}\label{l:I>III}
\begin{enumerate}
\item[(i)] If $|z_0|\le 1$, then  
 \[
F_{I}(\alpha, b,|z_0|^2)\ge F_{III}(b,|z_0|^2)
\]
wherever $*$-saddle point exists, i.e.
\begin{itemize}
    \item  $b\le 1$, $|z_0|\le 1$;
    \item  $b\in (1,\sqrt2]$, $\tfrac{b^2-b\sqrt{2-b^2}}{2}\le |z_0|^2\le 1$.
\end{itemize}

\item[(ii)] if $|z_0|^2> \max(1,b)$, $b\ge 1/\sqrt{2}$, then there exists a curve $z_2(b)$ such that  
 \[
F_{I}(\alpha, b,|z_0|^2)\ge F_{III}(b,|z_0|^2), \quad \max(1,b)\le |z_0|^2 \le z_2(b).
\]
If $|z_0|^2> z_2(b)$, then $*$-saddle point does not exist or 
 \[
F_{I}(\alpha, b,|z_0|^2)< F_{III}(b,|z_0|^2).
\]
\item[(iii)] if $|z_0|> 1$, but $b<1/\sqrt2$, then $*$-saddle point does not exist. 
\end{enumerate}
Therefore, the curve $\gamma_3$ on Figure 1 is 
\[
|z_0|^2=z_2(b), \quad b: z_2(b)<b.
\]
\end{lem}
\begin{proof} We start with (i). Notice that if $|z_0|^2\le b$ and so $v$-point exists, then the statement follows from (\ref{II,III_in}) and Lemma \ref{l:I>II}. It is easy to see also that the second derivative with respect to $t$ of $f_0(tI,v)$ at zero saddle-point has the form
\[
(f_0'')_t(0,0)=4\Big(\dfrac{1}{|z_0|^2}-1\Big),
\]
and hence it is positive for $|z_0|^2<1$. Therefore, this stationary point cannot be the point of local maximum.
Hence, if $|z_0|^2\le b$ and so $v$-point does not exist, the global maximum can be achieved only at $*$-saddle point (see Proposition \ref{p:main_saddle}) which implies (i). (iii) follows from Lemma \ref{l:*}.

It remains to prove (ii). Assume $|z_0|^2> 1$, $b\ge 1/\sqrt{2}$. Since we are interested in the case $|z_0|^2>b$ and $*$-point exists, according to Lemma
\ref{l:*}, we need to consider $b\in [\tfrac{1}{\sqrt2},\sqrt{\tfrac{5+\sqrt5}{5}}]$
\[
\max(1,b)\le |z_0|^2\le z_-(b)
\]
with $z_-$ of (\ref{3root}).
Since
\begin{equation}\label{b_bound}
 z_-(b)>b\Longrightarrow b\le b_0\approx 1.128,
\end{equation}
we are interested in $b\in [\tfrac{1}{\sqrt2},b_0]$.

In addition, we get
\[
\alpha\ge \alpha_-(b)
\]
with $\alpha_-$ of (\ref{alp_pm}) and, since $|z_0|> 1$ (see (\ref{al_in})),
\[
\alpha< 1-\dfrac{1}{b\sqrt2}.
\]
Define
\begin{align}\label{h}
    &W(\alpha,b)=F_I(\alpha,b,|z_0|^2)-F_{III}(\alpha,b,|z_0|^2)\\
    \notag &=-\alpha^2b^2-\dfrac{2\alpha |z_0|^2}{1-\alpha}+2\alpha b^2-\log(|z_0|^2(1-\alpha))\\ \notag
    &=b^2(2\alpha-5\alpha^2+4\alpha^3)-2\alpha-2\log(1-\alpha)-\log(1+2\alpha(1-\alpha)b^2).
\end{align}
Here we used 
\[
\dfrac{|z_0|^2}{1-\alpha}=1+2\alpha(1-\alpha)b^2.
\]
Consider now
\begin{align*}
    &\tfrac12 W'(\alpha,b)=b^2(1-5\alpha+6\alpha^2)-1+\dfrac{1}{1-\alpha}-\dfrac{1-2\alpha}{1+2\alpha(1-\alpha)b^2}\\
    &=\dfrac{\alpha}{(1-\alpha)(1+2\alpha(1-\alpha)b^2)}(1+b^2(-3+11\alpha-8\alpha^2)+2b^4(1-\alpha)^2(1-5\alpha+6\alpha^2))\\
    &=\dfrac{\alpha}{(1-\alpha)(1+2\alpha(1-\alpha)b^2)}\big(2(1-\alpha)(3\alpha-1)b^2+1\big)\big((1-\alpha)(2\alpha-1)b^2+1\big).
\end{align*}
According to (\ref{der}),
\[
2(1-\alpha)(3\alpha-1)b^2+1=1-b^2(6\alpha^2-8\alpha+2)\ge 0.
\]
Since $\alpha<1-\tfrac{1}{b\sqrt2}\le 1/2$, we have for $\alpha\in [\alpha_-(b),1-\tfrac{1}{b\sqrt2}]$
\begin{align*}
  (1-\alpha)(2\alpha-1)b^2+1\ge 1-b^2(1-\alpha_-(b))(1-2\alpha_-(b))=\dfrac{6-b^2-\sqrt{b^4+3b^2/2}}{9}>0  
\end{align*}
for $b<\sqrt2$.
Therefore, we obtain
\[
W'(\alpha,b)\ge 0, \quad \alpha\in [\alpha_-(b),1-\tfrac{1}{b\sqrt2}]
\]
and hence $W(\alpha,b)$ increases on $\alpha\in [\alpha_-(b),1-\tfrac{1}{b\sqrt2}]$.

In addition, if $\alpha=1-\tfrac{1}{b\sqrt2}$, then
\[
|z_0|^2=1, \quad b^2=\tfrac{1}{2(1-\alpha)^2},
\]
and so
\begin{align*}
W(1-\tfrac{1}{b\sqrt2},b)&=b^2(2\alpha-5\alpha^2+4\alpha^3)-2\alpha-\log(1-\alpha)\\
&=\dfrac{2\alpha-5\alpha^2+4\alpha^3}{2(1-\alpha)^2}-2\alpha-\log(1-\alpha)\\
&\ge \dfrac{2\alpha-5\alpha^2+4\alpha^3}{2(1-\alpha)^2}-2\alpha+\alpha+\dfrac{\alpha^2}{2}=\dfrac{\alpha^4}{2(1-\alpha)^2}\ge 0.
\end{align*}
Therefore, 
\begin{align}\label{I>III}
    W(\alpha,b)\ge 0, \quad \alpha\in [\alpha_1(b),1-\tfrac{1}{b\sqrt2}]
\end{align}
where $\alpha_1(b)=\alpha_-(b)$ if $W(\alpha_-(b),b)\ge 0$ and
$\alpha_1(b)\in [\alpha_-(b),1-\tfrac{1}{b\sqrt2}]$ such that
\[
W(\alpha_1(b),b)=0
\]
if $W(\alpha_-(b),b)<0$. 

According to (\ref{der}), it gives that (\ref{I>III}) holds
for any $b\in [1/\sqrt2, \sqrt2]$, $\max(1,b)\le |z_0|^2\le z_2(b)$ where 
\[
z_2(b)=2\alpha_1(b)(1-\alpha_1(b))^2b^2+1-\alpha_1(b),
\]
and the opposite inequality holds if $|z_0|^2> z_2(b)$, as desired.

Notice that one can compute numerically that starting from $b>b_1\approx 1.11$ we get $z_2(b)<b$, and so the interval 
$\max(1,b)\le |z_0|^2\le z_2(b)$ is empty.
\end{proof}

\subsection{Integral estimates}
Now we proceed to the integral estimates. Consider first the domain $\Omega_1$ where the $*$-saddle point dominates (see Lemmas \ref{l:I>II}-\ref{l:I>III}). In a standard way the integration domain in \eqref{IR SVD} can be restricted as follows
\begin{equation*}
	\CF_2(z_1,z_2) = \frac{2n^5}{\pi} \int\limits_{\Sigma_r} \Vanddet^2(T^2) t_1t_2 \times e^{nf(T, U, V, v)} d\mu(U) d\mu(V) dTd\bar vdv + O(e^{-nr/2}),
\end{equation*}
where
\begin{equation*}
	\Sigma_r = \left\{(T, U, V, v) \in \idom \mid \norm{T} + \abs{v} \le r\right\}.
\end{equation*}
The next step is to restrict the integration domain by $\frac{\log n}{\sqrt{n}}$-neighborhood of the *-saddle point.
To this end we need to expand $f$ near the $*$-saddle point $(t_*I,x_*)$:
\begin{lem}\label{lem:f(UT V^*) expansion}
	 Let $\tilde{T}$ be a $2 \times 2$ diagonal matrix such that $\normsized{\tilde{T}} \le \log n$ and $\tilde{v} = \tilde{x} + i\tilde{y}$ be a complex number with $\abs{\tilde{v}} \le \log n$. Then uniformly in $U$ and $V$
	\begin{equation} \label{f expansion}
	\begin{split}
	f(t_*I + \tfrac{1}{\sqrt{n}}\tilde{T}, U, V, x_* + \tfrac{1}{\sqrt{n}}\tilde{v}) = f_{0*} + n^{-1/2} \frac{\abs{z_0}^2 + t_*^2}{h_*} \tr (\cConjScl{z}_0\cMatPos + z_0\cMatPos^*) \\
    + n^{-1} \qform^{(I)}(\tilde{T}, U, V, \tilde{v}) + O\big(n^{-3/2}\log^3 n\big)
	\end{split}
	\end{equation}
	where
	\begin{align}
    \begin{split}\label{qform def}
    \qform^{(I)}(\tilde{T}, U, V, \tilde{v}) &= \frac{1}{2h_*}\tr \left[ -
    4(t_*^2 + bx_*)\tilde{T}^2 - 4t_*\tilde{T}P_1 - P_1^2 + 2(\abs{z_0}^2 + t_*^2)\cMatPos_U\cMatPos_V^* \right] \\
    &\quad{}+ \frac{1}{2h_*}\left[(2t_*\tr \tilde{T} + \tr P_1 )^2 + 4bt_*\tilde{x} \tr \tilde T + 2bx_*(\tr \tilde T)^2 + 2b^2\abs{\tilde{v}}^2\right] \\
    &\quad{}- \frac{1}{2h_*^2}\left[ 2h_*(t_*\tr \tilde{T} + x_*\tilde{x} ) + (\abs{z_0}^2 + t_*^2)\tr P_1 \right]^2 - \abs{\tilde{v}}^2;
    \end{split}\\
    \label{zeta_B def}
    \cMatPos&=\begin{pmatrix}
        \zeta_1&0\\
        0&\zeta_2
    \end{pmatrix},\quad \cMatPos_B = B^*\cMatPos B; \\
    \notag
    P_1 &= \cConjScl{z}_0\cMatPos_U + z_0\cMatPos_V^*,
	\end{align}
 and $f_{0*}=f_0(t_*I, x_*)$.
\end{lem}

\begin{proof}
If $Q = U(t_*I + n^{-1/2}\tilde{T})V^*$ then the matrix $A$ \eqref{A def} has the form
	\begin{equation*}
		A = \begin{pmatrix}
			U 	& 0 \\
			0 	& V
		\end{pmatrix}\left(A_0 + \frac{1}{\sqrt{n}}A_1\right)\begin{pmatrix}
		U^*	& 0 \\
		0	& V^*
	\end{pmatrix},
\end{equation*}
where
\begin{equation*}\label{A_0,A_1 def}
A_0 = \begin{pmatrix}
-z_0I 		& t_*I \\
-t_*I 	& -\bar z_0I
\end{pmatrix}, \quad
A_1 = \begin{pmatrix}
-\cMatPos_U 		& \tilde{T} \\
-\tilde{T} 	& -\cMatPos_V^*
\end{pmatrix}.
\end{equation*}
One gets
\begin{equation}\label{log det A}
\begin{split}
\log \det A &= \log \det A_0 + \log \det A_0^{-1}A = \log \det A_0 + \tr \log (1 + n^{-1/2}A_0^{-1}A_1) \\
&= \log \det A_0 + \frac{1}{\sqrt{n}}\tr A_0^{-1}A_1 - \frac{1}{2n}\tr (A_0^{-1}A_1)^2 + O\left(\frac{\log^3 n}{\sqrt{n^3}}\right)
\end{split}
\end{equation}
uniformly in $U$ and $V$. Moreover,
\begin{align*}
A_0^{-1} &= \frac{1}{\abs{z_0}^2 + t_*^2}\begin{pmatrix}
-\cConjScl{z}_0I & -t_*I \\
t_*I		     & -z_0I
\end{pmatrix},\\
A_0^{-1}A_1 &= \frac{1}{\abs{z_0}^2 + t_*^2}\begin{pmatrix}
\cConjScl{z}_0\cMatPos_U + t_*\tilde{T}	& -\cConjScl{z}_0\tilde{T} + t_*\cMatPos_V^* \\
-t_*\cMatPos_U + z_0\tilde{T}		& t_*\tilde{T} + z_0\cMatPos_V^*
\end{pmatrix},\\
\tr A_0^{-1}A_1 &=\frac{1}{\abs{z_0}^2 + t_*^2}\tr\left[ 2t_* \tilde{T} + P_1 \right],\\
\tr (A_0^{-1}A_1)^2 &=\frac{1}{\left(\abs{z_0}^2 + t_*^2\right)^2}\tr\Bigl[ 2(t_*^2 - \abs{z_0}^2) \tilde{T}^2 + 4t_*P_1\tilde{T}\\ &\quad + \cConjScl{z}_0^2\cMatPos_U^2 + z_0^2(\cMatPos_V^*)^2 
- 2t_*^2\cMatPos_U \cMatPos_V^* \Bigr]
\end{align*}
where $P_1 = \cConjScl{z}_0\cMatPos_U + z_0\cMatPos_V^*$. \eqref{log det A} implies
\begin{multline}\label{det A expansion}
    \det A = (\abs{z_0}^2 + t_*^2)^2\left( 1 + \frac{1}{\sqrt{n}} \tr A_0^{-1}A_1 \right.\\
    \left.+ \frac{1}{2n}\big(\tr A_0^{-1}A_1)^2 - \tr (A_0^{-1}A_1)^2\big) \right)
    + O\left(\frac{\log^3 n}{\sqrt{n^3}}\right).
\end{multline}
Further,
\begin{equation}\label{x t1 t2}
\begin{split}
    xt_1t_2 &= x_*t_*^2 + \frac{1}{\sqrt{n}}(t_*^2\tilde{x} + x_*t_*\tr \tilde{T}) \\
    &\quad{}+ \frac{1}{2n}\left(2t_*\tilde{x}\tr \tilde{T} + x_*(\tr \tilde{T})^2 - x_*\tr \tilde{T}^2\right) + O\left(\frac{\log^3 n}{\sqrt{n^3}}\right).
\end{split}
\end{equation}
Equations \eqref{det A expansion} and \eqref{x t1 t2} yield for $Q = U(t_*I + n^{-1/2}\tilde{T})V^*$, $v=x_*+n^{-1/2}\tilde v$ 
\begin{equation}\label{h expan}
    \begin{split}
    &h(Q,v)= h_* + \frac{1}{\sqrt{n}}\left[ (\abs{z_0}^2 + t_*^2)a_1 + 2bt_*^2\tilde{x} + 2bx_*t_*\tr \tilde{T} + 2b^2x_*\tilde{x} + 2b^2y_0\tilde{y} \right] \\
    &+ \frac{1}{2n}\left[ a_1^2 - a_2 + 4bt_*\tilde{x}\tr \tilde{T} + 2bx_*(\tr \tilde{T})^2 - 2bx_*\tr \tilde{T}^2 + 2b\tilde{x}^2 + 2b\tilde{y}^2 \right] + O\left(\frac{\log^3 n}{\sqrt{n^3}}\right),
    \end{split}
\end{equation}
where $h$ is defined in \eqref{h_v}, and
\begin{align}
    \notag
    a_1 &= (\abs{z_0}^2 + t_*^2)\tr A_0^{-1}A_1 = \tr\left[ 2t_* \tilde{T} + P_1 \right],\\
    \label{a_2 def}
    \begin{split}
    a_2 &= (\abs{z_0}^2 + t_*^2)^2\tr (A_0^{-1}A_1)^2 \\
    &= \tr\Bigl[ 2(t_*^2 - \abs{z_0}^2) \tilde{T}^2 + 4t_*P_1\tilde{T} + \cConjScl{z}_0^2\cMatPos_U^2 + z_0^2(\cMatPos_V^*)^2 - 2t_*^2\cMatPos_U \cMatPos_V^* \Bigr].
    \end{split}
\end{align}
Using the equations \eqref{system}, the equation \eqref{h expan} can be transformed to
\begin{align}\label{h expan1}
    &h(Q,v) = h_* + \frac{1}{\sqrt{n}}\left[ 2h_*t_*\tr \tilde{T} + 2h_*x_*\tilde{x}  + (\abs{z_0}^2 + t_*^2)\tr P_1 \right] \\ \notag
    &+ \frac{1}{2n}\left[ a_1^2 - a_2 + 4bt_*\tilde{x}\tr \tilde{T} + 2bx_*(\tr \tilde{T})^2 - 2bx_*\tr \tilde{T}^2 + 2b\tilde{x}^2 + 2b\tilde{y}^2 \right] + O\left(\frac{\log^3 n}{\sqrt{n^3}}\right).
\end{align}
Substituting \eqref{h expan1} into \eqref{f def} and expanding the logarithm, we obtain
\begin{align*}
&f\left(t_*I + \tfrac{1}{\sqrt{n}}\tilde{T}, U, V, x_* + \tfrac{1}{\sqrt{n}}\tilde{v}\right) = f_{0*} + n^{-1/2} \frac{\abs{z_0}^2 + t_*^2}{h_*} \tr (\cConjScl{z}_0\cMatPos + z_0\cMatPos^*)  \\   &+ n^{-1} \biggl\lbrace - \tr \tilde{T}^2 
    - \tilde{x}^2 - \tilde{y}^2- \frac{1}{2h_*^2}\left[ 2h_*t_*\tr \tilde{T} + 2h_*x_*\tilde{x} + (\abs{z_0}^2 + t_*^2)\tr P_1 \right]^2\\
    &+ \frac{1}{2h_*}\left[ a_1^2 - a_2 + 4bt_*\tilde{x}\tr \tilde{T} + 2bx_*(\tr \tilde{T})^2 - 2bx_*\tr \tilde{T}^2 + 2b\tilde{x}^2 + 2b\tilde{y}^2 \right] \biggr\rbrace + O\Big(\tfrac{\log^3 n}{n^{3/2}}\Big).
\end{align*}
The last expansion, \eqref{a_2 def}, and \eqref{h_*} imply \eqref{f expansion}.
\end{proof}
We also need
\begin{lem}\label{lem:est for Re f}
	Let the $*$-saddle-point $(t_*, t_*, x_*, 0)$ defined by Proposition \ref{p:main_saddle} be a unique global maximum point of the function $f_0(T,v)=f_0(t_1, t_2, x, y)$. Let 
\begin{equation*}
    \tilde{f}(T, U, V, v) = f(T, U, V, v) - f_*,
\end{equation*}
  where  $f_* = f(t_*I, I, I, x_*)$ with $f$ of (\ref{f def}).
    Then for sufficiently large~$n$
	\begin{equation*} 
		\max_{\frac{\log n}{\sqrt{n}} \le \norm{T - t_*I} + \abs{v - x_*} \le r} \tilde{f}(T, U, V, v) \le -C\frac{\log^2 n}{n}
	\end{equation*}
	uniformly in $U$ and $V$.
\end{lem}
\begin{proof}
	First let us check that the first 
    derivatives of $f_r$ are bounded in the $\delta$-neighborhood of the manifold $(t_*I, U, V, x_*)$, $U, V \in U(2)$, where $f_r$ is defined in \eqref{f_r def} and $\delta$ is $n$-independent. Indeed, since $h$ and $h_0$ are polynomials 
	\begin{align*}
		\abs{\frac{1}{\sqrt{n}}\frac{\partial f_r}{\partial s}} = \abs{\frac{\partial (f - f_0)}{\partial s}} = \abs{\frac{\partial (\log h - \log h_0)}{\partial s}} 		\le \abs{\frac{1}{h_0} \cdot \frac{\partial h_0}{\partial s} - \frac{1}{h} \cdot \frac{\partial h}{\partial s}} \le \frac{C}{\sqrt{n}},
	\end{align*}
	where $s$ is either $t_1$, $t_2$ $x$ or $y$. Let $T_E$ be a real diagonal $2 \times 2$ matrix of unit norm and let $v_E \in \Compl$ be a number on the unit circle. Then for any $T_E$ and $v_E$ and for $\frac{\log n}{\sqrt{n}} \le \sigma \le \delta$ we have
	\begin{equation*}
		\begin{split}
			\der{}{\sigma} \tilde{f}(t_*I + \sigma T_E, U, V, x_* + \sigma v_E) &= \langle \nabla_{T, x, y} f_0(t_*I + \sigma T_E, x_* + \sigma v_E), \EcVec_E \rangle \\
			&\quad{}+ n^{-1/2} \langle \nabla_{T, x, y} f_r(t_*I + \sigma T_E, U, V, x_* + \sigma v_E), \EcVec_E \rangle \\
			&= \langle \nabla_{T, x, y} f_0(t_*I + \sigma T_E, x_* + \sigma v_E), \EcVec_E \rangle + O(n^{-1/2}),
		\end{split}
	\end{equation*}
	where $\EcVec_E$ denotes a vector $(t_{E1}, t_{E2}, x_E, y_E)$ and $\langle \cdot, \cdot \rangle$ is a standard real scalar product. Expanding the scalar product by Taylor formula and considering that $\nabla_{T, x, y} f_0(t_*I, x_*) = 0$, we obtain
	\begin{equation*}
		\begin{split}
			\der{}{\sigma} \tilde{f}(t_*I + \sigma T_E, U, V, x_* + \sigma v_E) &= \sigma\langle f_0''(t_*I, x_*)\EcVec_E, \EcVec_E \rangle + r_1 + O(n^{-1/2}),
		\end{split}
	\end{equation*}
	where $f_0''$ is a matrix of second order derivatives of $f_0$ w.r.t.\ $T$, $x$ and $y$ and $\abs{r_1} \le C\sigma^2$. $f_0''(t_*I, x_*)$ is non-negative definite, because $(t_*I, x_*)$ is a global maximum. Putting $\cMatPos = 0$ into \eqref{f expansion} and \eqref{qform def} one obtains
    \begin{equation*}
    f_0''(t_*I, x_*) = -2\begin{pmatrix}
    2t_*^2 + \frac{bx_*}{h_*} & 2t_*^2 - \frac{2t_*^2 + bx_*}{h_*} & t_*\left(2x_* - \frac{b}{h_*}\right) & 0 \\
    2t_*^2 - \frac{2t_*^2 + bx_*}{h_*} & 2t_*^2 + \frac{bx_*}{h_*} & t_*\left(2x_* - \frac{b}{h_*}\right) & 0 \\
    t_*\left(2x_* - \frac{b}{h_*}\right) & t_*\left(2x_* - \frac{b}{h_*}\right) & 1 + 2x_*^2 - \frac{b^2}{h_*} & 0 \\
    0 & 0 & 0 & 1 - \frac{b^2}{h_*}
    \end{pmatrix}.
    \end{equation*} 
    A straightforward check shows that $\det f_0''(t_*I, x_*) > 0$. Hence $f_0''(t_*I, x_*)$ is negative definite and $\der{}{\sigma} (t_*I + \sigma T_E, U, V, x_* + \sigma v_E)$ is negative and
	\begin{equation}\label{Re f: nbh est}
	\begin{split}
	\max_{\frac{\log n}{\sqrt{n}} \le \norm{T - t_*I} + \abs{v - x_*} \le \delta} \tilde{f}(T, U, V, v) &= \max_{\norm{T - t_*I} + \abs{v - x_*} = \frac{\log n}{\sqrt{n}}} \tilde{f}(T, U, V, v) \\
	&\le f(t_*I, U, V, x_*) - C\frac{\log^2 n}{n} - f_*.
	\end{split}
	\end{equation}
	Notice that $f_r$ is bounded from above uniformly in $n$ and $(t_*, t_*, x_*, 0)$ is a point of  global maximum of the function $f_0$. These facts imply that $\delta$ in \eqref{Re f: nbh est} can be replaced by $r$
	\begin{equation*}
		\max_{\frac{\log n}{\sqrt{n}} \le \norm{T - t_*I} + \abs{v - x_*} \le r} \tilde{f}(T, U, V, v) \le f(t_*I, U, V, x_*) - f_* - C\frac{\log^2 n}{n}.
	\end{equation*}
	It remains to deduce from Lemma \ref{lem:f(UT V^*) expansion} that $f(t_*I, U, V, x_*) - f_* = O(n^{-1})$ uniformly in $U$ and~$V$.
\end{proof}
Lemma \ref{lem:est for Re f} yields
\begin{equation*}
	\CF_2(z_1,z_2) = \frac{2n^5e^{nf_*}}{\pi} \Bigg(\int\limits_{\stpointsnbh(t_*I, x_*)} \Vanddet^2(T^2) t_1t_2 e^{n\tilde{f}(T, U, V, v)} d\mu(U) d\mu(V) dT d\cConjScl{v}dv + O(e^{-C_1\log^2 n})\Bigg),
\end{equation*}
where by $\stpointsnbh(\hat T, \hat v)$ we denote a $\frac{\log n}{\sqrt{n}}$-neighborhood of the point $(\hat{T}, \hat{v})$, i.e.
\begin{equation}\label{stpoinnbh def}
\stpointsnbh(\hat T, \hat v) = \left\{(T, U, V, v) \in \idom \mid \normsized[\big]{T - \hat T} \le \frac{\log n}{\sqrt{n}}, \abs{v - \hat v} \le \frac{\log n}{\sqrt{n}}\right\},
\end{equation}
Changing the variables $T = t_*I + \tfrac{1}{\sqrt{n}}\tilde{T}$ and $v = x_* + \frac{1}{\sqrt{n}}\tilde{v}$, and expanding $\Vanddet^2(T^2) t_1t_2$ and the function~$f$ according to Lemma~\ref{lem:f(UT V^*) expansion}, we obtain
\begin{equation}\label{n-indep int}
\CF_2(z_1,z_2) = \frac{8t_*^4}{\pi}\mathsf{k}_n \int\limits_{\sqrt{n}\stpointsnbh(0)} \Vanddet^2(\tilde{T}) \exp\left\lbrace\qform(\tilde{T}, U, V, \tilde{v})\right\rbrace  d\mu(U) d\mu(V) d\tilde{T} d\cConjScl{\tilde{v}} d\tilde{v}(1 + o(1)),
\end{equation}
where
\begin{equation}\label{K_n def}
\mathsf{k}_n = n^2\exp\left\lbrace nf_{0*} + \sqrt{n}\frac{\abs{z_0}^2 + t_*^2}{h_*}\tr \left(\cConjScl{z}_0\cMatPos + z_0\cMatPos^*\right) \right\rbrace.
\end{equation}
Let us change the variables $V = WU$. Taking into account that the Haar measure is invariant w.r.t.\ shifts we get
\begin{equation*}
    \CF_2(z_1,z_2) = \frac{8t_*^4}{\pi}\mathsf{k}_n \int\limits_{\R^4} d\tilde{T}d\tilde{x}d\tilde{y} \int\limits_{U(2)} d\mu(U) \int\limits_{U(2)} d\mu(W) \Vanddet^2(\tilde{T}) \exp\left\{ \qform_1(\tilde{T}, U, W, \tilde{v}) \right\}(1 + o(1)),
\end{equation*}
where
\begin{align*}
    \qform_1(\tilde{T}, U, W, \tilde{v}) &= \frac{1}{2h_*}\tr \left[ -4(t_*^2 + bx_*)U\tilde{T}^2U^* - 4t_*U\tilde{T}U^*P - P^2 + 2(\abs{z_0}^2 + t_*^2)\cMatPos\cMatPos_W^* \right] \\
    &+ \frac{1}{2h_*}\left[( 2t_* \tr \tilde{T} + \tr P )^2 + 4bt_*\tilde{x} \tr \tilde T + 2bx_*(\tr \tilde T)^2 + 2b^2\abs{\tilde{v}}^2\right] - \abs{\tilde{v}}^2 \\
    &- \frac{1}{2h_*^2}\left[ 2t_*h_*\tr \tilde{T} + 2x_*h_*\tilde{x} + (\abs{z_0}^2 + t_*^2)\tr P \right]^2
\end{align*} 
with

\begin{align*}
    P = \cConjScl{z}_0\cMatPos + z_0\cMatPos_W^*.
\end{align*}
The next step is to change the variables $(\tilde{T}, U) \to H$ such that $H = U\tilde{T}U^*$. The Jacobian is $\frac{1}{(2\pi)}\Vanddet^{-2}(\tilde{T})$ (see e.g.\ \cite{Hu:63}). Thus
\begin{equation*}
	\CF_2(z_0+\tfrac{\zeta_1}{\sqrt n},z_0+\tfrac{\zeta_2}{\sqrt n}) = \frac{4t_*^4}{\pi^2}\mathsf{k}_n \int\limits_{\herm_2} dH \int\limits_{\R^2} d\tilde{x}d\tilde{y} \int\limits_{U(2)} d\mu(W) \exp\left\{ \qform_2(H, W, \tilde{v}) \right\}(1 + o(1)),
\end{equation*}
where $\herm_2$ is a space of hermitian $2 \times 2$ matrices and
\begin{align*}
	dH &= d(H)_{11}d(H)_{22} d\Re (H)_{12} d\Im (H)_{12},\\
    \qform_2(H, W, \tilde{v}) &= \frac{1}{2h_*} \left[ -\tr \left( 2\sqrt{t_*^2 + bx_*}H + \frac{t_*}{\sqrt{t_*^2 + bx_*}}P \right)^2  + 2(\abs{z_0}^2 + t_*^2)\tr \cMatPos\cMatPos_W^* \right] \\
    &+ \frac{1}{2h_*}\left[(2t_*\tr H + \tr P)^2 + 4bt_*\tilde{x} \tr H + 2bx_*(\tr H)^2 + 2b^2\abs{\tilde{v}}^2- \frac{bx_*}{t_*^2 + bx_*}\tr P^2\right]\\
    &- \frac{1}{2h_*^2}\left[ 2t_*h_*\tr H + 2x_*h_*\tilde{x} + (\abs{z_0}^2 + t_*^2)\tr P \right]^2- \abs{\tilde{v}}^2 .
\end{align*}
Shifting the variables $H \to H - \frac{t_*}{2(t_*^2 + bx_*)}P$ and moving integration back to the real axis, one has
\begin{equation*}
	\CF_2(z_0+\tfrac{\zeta_1}{\sqrt n},z_0+\tfrac{\zeta_2}{\sqrt n}) = \frac{4t_*^4}{\pi^2}\mathsf{k}_n \int\limits_{\herm_2} dH \int\limits_{\R^2} d\tilde{x}d\tilde{y} \int\limits_{U(2)} d\mu(W) \exp\left\{ \qform_3(H, W, \tilde{v}) \right\}(1 + o(1)),
\end{equation*}
where
\begin{align*}
    \qform_3(H, W, \tilde{v}) &= \frac{1}{2h_*}\tr \left[ - \frac{bx_*}{t_*^2 + bx_*}P^2 + 2(\abs{z_0}^2 + t_*^2)\cMatPos\cMatPos_W^* \right] \\
    &-\frac{4(t_*^2 + bx_*)}{h_*}\abs{(H)_{12}}^2 - \left(1 - \frac{b^2}{h_*}\right)\tilde{y}^2 - \frac{1}{2}\left\langle B\UcVec, \UcVec \right\rangle - \left\langle \QcVec, \UcVec \right\rangle \\
    &+ \frac{bx_*(h_*t_*^2 + 2(h_* - \abs{z_0}^4)bx_*)}{4h_*^2(t_*^2 + bx_*)^2}(\tr P)^2
\end{align*}
with
\begin{align*}
B &= 2\begin{pmatrix}
2t_*^2 + \frac{bx_*}{h_*} & 2t_*^2 - \frac{2t_*^2 + bx_*}{h_*} & t_*\left(2x_* - \frac{b}{h_*}\right) \\
2t_*^2 - \frac{2t_*^2 + bx_*}{h_*} & 2t_*^2 + \frac{bx_*}{h_*} & t_*\left(2x_* - \frac{b}{h_*}\right) \\
t_*\left(2x_* - \frac{b}{h_*}\right) & t_*\left(2x_* - \frac{b}{h_*}\right) & 1 + 2x_*^2 - \frac{b^2}{h_*}
\end{pmatrix}, \\ \UcVec &= \begin{pmatrix}
    (H)_{11} \\
    (H)_{22} \\
    \tilde{x}
\end{pmatrix}, \qquad \QcVec = \frac{b \tr P}{h_*(t_*^2 + bx_*)} \begin{pmatrix}
    x_*t_*(2\abs{z_0}^2 - 1) \\
    x_*t_*(2\abs{z_0}^2 - 1) \\
    t_*^2 + 2\abs{z_0}^2x_*^2
\end{pmatrix}.
\end{align*} 
The Gaussian integration over $H$ and $\tilde{v}$ implies
\begin{align}\label{last asympt}
&\CF_2(z_0+\tfrac{\zeta_1}{\sqrt n},z_0+\tfrac{\zeta_2}{\sqrt n}) = C(\zeta_1, \zeta_2) \int\limits_{U(2)} d\mu(W) \\ \notag
&\quad{}\times\exp\left\{ \frac{1}{h_*}\left[ -\frac{bx_*\abs{z_0}^2}{t_*^2 + bx_*} + (\abs{z_0}^2 + t_*^2) \right] \tr\cMatPos W^*\cMatPos^* W\right\} (1 + o(1)) \\ \notag
&=C(\zeta_1, \zeta_2) \int\limits_{U(2)} d\mu(W) \exp\left\{ \left[ 1 - \frac{b^2}{h_*} \right] \tr\cMatPos W^*\cMatPos^* W\right\} (1 + o(1)),
\end{align}
where 
\begin{align}\label{preint multiplier}
    C(\zeta_1, \zeta_2) &= \frac{\pi t_*^4 h_* \mathsf{k}_n }{t_*^2 + bx_*} \sqrt{\frac{8h_*}{(h_* - b^2)\det B}}
    \exp\left\lbrace - \frac{bx_*}{2h_*(t_*^2 + bx_*)}\tr (\cConjScl{z}_0^2\cMatPos^2 + z_0^2 (\cMatPos^*)^2) \right\rbrace\\ \notag
    &\quad{}\times \exp\left\lbrace\frac{1}{2}\left\langle B^{-1}\QcVec, \QcVec \right\rangle + \frac{bx_*(h_*t_*^2 + 2(h_* - \abs{z_0}^4)bx_*)}{4h_*^2(t_*^2 + bx_*)^2}(\tr P)^2\right\rbrace. 
   \end{align}
Straightforwardly substituting  the expressions \eqref{alp_expr}--\eqref{cub_eq}, \eqref{h_*} for $t_*^2$, $h_*$, and then for $\abs{z_0}^2$ into \eqref{preint multiplier} one can get
\begin{equation*}
    \exp\left\lbrace\frac{1}{2}\left\langle B^{-1}\QcVec, \QcVec \right\rangle + \frac{bx_*(h_*t_*^2 + 2(h_* - \abs{z_0}^4)bx_*)}{4h_*^2(t_*^2 + bx_*)^2}(\tr P)^2\right\rbrace 
    = \exp\left\lbrace \gamma \, (\tr P)^2/4 \right\rbrace.
\end{equation*}
where
\begin{equation}\label{gamma}
    \gamma= \frac{2b^2(1 - (1 - 4\alpha + 2\alpha^2)b^2)}{(1 - (1 - 4\alpha + 3\alpha^2)b^2)(1 + (2\alpha - 2\alpha^2)b^2)}
\end{equation}
and $\alpha$ is as in Lemma \ref{l:*}.

For computing the integral over the unitary group, we use the well-known Harish Chan\-dra/It\-syk\-son--Zuber formula 
\begin{prop}[see, e.g., \cite{Me:91}, A5]\label{pr:H-C/I--Z formula}
	Let $A$ and $B$ be normal $d \times d$ matrices with distinct eigenvalues $\{a_j\}_{j = 1}^d$ and $\{b_j\}_{j = 1}^d$ respectively. 
	Then
	\begin{equation*} 
		\int\limits_{U(d)} \exp\{t\tr AU^*BU\}d\mu(U) = \bigg(\prod\limits_{j = 1}^{d - 1} j!\bigg) \frac{\det\{\exp(ta_jb_k)\}_{j,k = 1}^d}{t^{(d^2 - d)/2}\Vanddet(A)\Vanddet(B)},
	\end{equation*}
	where $t$ is some constant, $\mu$ is a Haar measure, and $\Vanddet(A) = \prod\limits_{j > k}(a_j - a_k)$.
\end{prop}
An application of the formula to  \eqref{last asympt} gives
\begin{equation*} 
	\CF_2(z_0+\tfrac{\zeta_1}{\sqrt n},z_0+\tfrac{\zeta_2}{\sqrt n}) = C(\zeta_1, \zeta_2) \frac{\det \left\lbrace \exp\left[\left( 1 - \frac{b^2}{h_*} \right) \zeta_j\cConjScl{\zeta}_k\right]\right\rbrace_{j,k = 1}^2}{\left( 1 - \frac{b^2}{h_*} \right)\abs{\Vanddet(\cMatPos)}^2} (1 + o(1)),
\end{equation*}
where 
\begin{equation*}
    C(\zeta_1, \zeta_2) = \frac{\pi t_*^4 h_* \mathsf{k}_n}{t_*^2 + bx_*} \sqrt{\frac{8h_*}{(h_* - b^2)\det B}}
    \exp\left\lbrace - \frac{b^2}{2h_*^2}\tr (\cConjScl{z}_0^2\cMatPos^2 + z_0^2 (\cMatPos^*)^2) + \gamma \, (\tr P)^2/4 \right\rbrace
\end{equation*}
with $\gamma$ of (\ref{gamma}). Now using
\begin{align*}
&\tr P=z_0(\bar\zeta_1+\bar\zeta_2)+\bar z_0(\zeta_1+\zeta_2),\\
&(z_0(\bar\zeta_1+\bar\zeta_2)+\bar z_0(\zeta_1+\zeta_2))^2-2(\bar z_0\zeta_1+z_0\bar\zeta_1)^2-2(\bar z_0\zeta_2+z_0\bar\zeta_2)^2=-4(\Re(\bar z_0(\zeta_1-\zeta_2)))^2 ,
\end{align*}
and denoting $\beta=1 - \frac{b^2}{h_*}$, we get (\ref{*-behav}). The equation on $\beta$ follows from 
\[
\beta=1-\dfrac{b^2(1-\alpha)}{|z_0|^2}\Longrightarrow \alpha=-\dfrac{(1-\beta) |z_0|^2}{b^2}+1
\]
and (\ref{cub_eq}). Here we also used (\ref{h_*}).

Suppose now $(b, |z_0|^2)\in \Omega_2$ where the main contribution is given by $v$-saddle point. Writing
\[
v_0 = r_0e^{i\varphi},\quad r_0=\sqrt{1 - \abs{z_0}^4 / b^2}, \quad \tilde{v} = \tilde{r}e^{i\varphi},
\]
changing the variables $T = \tfrac{1}{\sqrt{n}}\tilde{T}$ and $v = (r_0 + \frac{1}{\sqrt{n}}\tilde{r})e^{i\varphi}$ and proceeding similarly to Lemmas \ref{lem:f(UT V^*) expansion}-\ref{lem:est for Re f}, we obtain
\begin{equation}\label{n-indep int II}
\CF_2(z_1,z_2) = \frac{2r_0\mathsf{k}_n^{II}}{\pi} \int\limits_{\sqrt{n}\stpointsnbh(0)} \Vanddet^2(\tilde{T}^2)\tilde{t}_1\tilde{t}_2 \exp\left\lbrace\qform^{II}(\tilde{T}, U, V, \tilde{v})\right\rbrace  d\mu(U) d\mu(V) d\tilde{T} d\tilde{r} d\varphi(1 + o(1)),
\end{equation}
where
\begin{align*}
    \qform^{II}(\tilde{T}, U, V, \tilde{r}e^{i\varphi}) &= \frac{1}{2b^2}\tr \left[ -2(br_0\cos\varphi + b^2 - \abs{z_0}^2)\tilde{T}^2 - \cConjScl{z}_0^2\cMatPos^2 - z_0^2(\cMatPos^*)^2 \right] \\
        &+ \frac{1}{2b^2}\left[(\tr P_1)^2 + 2br_0\cos\varphi(\tr \tilde T)^2 \right] - \frac{1}{2b^4}\left[ 2b^2r_0\tilde{r} + \abs{z_0}^2\tr P_1 \right]^2,
    \end{align*}
\begin{equation}\label{K_n def II}
\mathsf{k}_n^{II} = n^{1/2}\exp\left\lbrace n\left(-1 + \frac{\abs{z_0}^4}{b^2} + \log b^2\right) + \sqrt{n}\frac{\abs{z_0}^2}{b^2}\tr \left(\cConjScl{z}_0\cMatPos + z_0\cMatPos^*\right) \right\rbrace.
\end{equation}
Taking the integral w.r.t.\ 
$\tilde{r}$ one gets
\begin{equation*}
\begin{split}
	\CF_2 &= C^{II}(\zeta_1, \zeta_2)\int\limits_{0}^{+\infty}d\tilde t_1 \int\limits_{0}^{+\infty}d\tilde t_2\int\limits_{0}^{2\pi}d\varphi\,\Vanddet^2(\tilde{T}^2)\tilde{t}_1\tilde{t}_2(1 + o(1)) \\
    &\quad{}\times\exp\left\lbrace \frac{1}{b^2}\left[ -(br_0\cos\varphi + b^2 - \abs{z_0}^2)\tr \tilde{T}^2 + br_0\cos\varphi(\tr \tilde T)^2 \right] \right\rbrace,
\end{split}
\end{equation*}
where
\begin{equation*}
    C^{II}(\zeta_1, \zeta_2) = \mathsf{k}_n^{II}\sqrt{\frac{2}{\pi}}\exp\left\lbrace \frac{1}{2b^2}\left[(\tr \cConjScl{z}_0\cMatPos + \tr z_0\cMatPos^*)^2 - \cConjScl{z}_0^2\tr\cMatPos^2 - z_0^2\tr(\cMatPos^*)^2\right]  \right\rbrace.
\end{equation*}
Since $1/2b^2=p/4$ according to (\ref{fin_p}), this implies (\ref{v-behav}).

Similarly, if $(b,|z_0|^2)\in \Omega_3$ where zero saddle point is dominant, 
changing the variables $T = \tfrac{1}{\sqrt{n}}\tilde{T}$ and $v = \tfrac{\tilde v}{\sqrt{n}}$, one obtains
\begin{equation*}
\CF_2(z_1,z_2) = \frac{2\mathsf{k}_n^{III}}{\pi} \int\limits_{\sqrt{n}\stpointsnbh(0)} \Vanddet^2(\tilde{T}^2)\tilde{t}_1\tilde{t}_2 \exp\left\lbrace\qform^{III}(\tilde{T}, U, V, \tilde{v})\right\rbrace  d\mu(U) d\mu(V) d\tilde{T} d\overline{\tilde v} d\tilde v(1 + o(1)),
\end{equation*}
where
\begin{equation*}
        \qform^{III}(\tilde{T}, U, V, \tilde{v}) = \frac{1}{2\abs{z_0}^4}\tr \left[ -2(\abs{z_0}^4 - \abs{z_0}^2)\tilde{T}^2 - \cConjScl{z}_0^2\cMatPos^2 - z_0^2(\cMatPos^*)^2 \right] 
        + \frac{b^2}{\abs{z_0}^4} \abs{\tilde{v}}^2 - \abs{\tilde{v}}^2,
    \end{equation*}
    and
\begin{equation}\label{K_n def III}
\mathsf{k}_n^{III} = \exp\left\lbrace n\log \abs{z_0}^4 + \frac{\sqrt{n}}{\abs{z_0}^2}\tr \left(\cConjScl{z}_0\cMatPos + z_0\cMatPos^*\right) \right\rbrace.
\end{equation}
Taking integration with respect to $\tilde t_1,\tilde t_2>0$ and and with respect to $\tilde v$ we get
\begin{equation*}
	\CF_2(z_0+\tfrac{\zeta_1}{\sqrt n},z_0+\tfrac{\zeta_2}{\sqrt n})=  \frac{\mathsf{k}_n^{III}(1 + o(1))}{\left(1 - \frac{1}{\abs{z_0}^2}\right)^4\left(1 - \frac{b^2}{\abs{z_0}^4}\right)}\exp\left\lbrace -\frac{1}{2\abs{z_0}^4}\left[ \cConjScl{z}_0^2\tr\cMatPos^2 + z_0^2\tr(\cMatPos^*)^2\right]  
 \right\rbrace,
\end{equation*}
which implies (\ref{0-behav}).

Theorem 2 follows from the consideration above for $b=0$ (see (\ref{inf_p})) with minor modifications. 


\bibliography{SpGinDetCorr}

\end{document}